\documentclass[12pt, draftclsnofoot, onecolumn]{IEEEtran}
\usepackage{hyperref}
\usepackage{amsmath,amssymb}
\usepackage{latexsym}
\usepackage{CJK}
\usepackage{cite}   % 文献引用

\usepackage{verbatim}  % \begin{comment} ... \end{comment} 包含被注释的文字
\usepackage{graphicx}  % 插入图片
\usepackage{mathrsfs}	% 花体字母
\usepackage{algorithm} %format of the algorithm
\usepackage{algorithmic} %format of the algorithm
\usepackage{booktabs}
\usepackage{textcomp}  %千分号\textperthousand 摄氏度：$^\circ$C \par
\usepackage{multirow}  % 多行合并
\usepackage{lettrine}   % 首字下沉
\usepackage{graphicx}  % 双栏模板插入通栏图片
\usepackage{bm}  % 希腊字母加粗 Eg. \bf{\gamma}
\usepackage{amsthm} %thm / proof环境等数学环境
\usepackage{subfigure} % 图片并排显示

\usepackage[usenames,dvipsnames]{color}
\usepackage{colortbl}
\definecolor{mygray}{gray}{.9}

\newtheorem{lem}{Lemma}
\newtheorem{defn}{Definition}
\newcommand{\myvec}[1]%
   {\stackrel{\raisebox{-2pt}[0pt][0pt]{\small$\rightharpoonup$}}{#1}}

\begin{document}
%\title{Location-aware Low Complexity ICI Reduction in MIMO-OFDM Downlinks for High-speed Railway Communication Systems with Distributed Antennas}

\title{Location-aware ICI Reduction in MIMO-OFDM Downlinks for High-speed Railway Communication Systems}

\author{\IEEEauthorblockN{Jiaxun Lu, Xuhong Chen, Shanyun Liu, Pingyi Fan}

\IEEEauthorblockA{
Tsinghua National Laboratory for Information Science and Technology(TNList),\\
Department of Electronic Engineering, Tsinghua University, Beijing, China.\\
Email: \{lujx14, chenxh13, sy-liu14\}@mails.tsinghua.edu.cn, fpy@mail.tsinghua.edu.cn}
}

\maketitle

% 设置图片位置：
\graphicspath{{Figures/}}

\begin{abstract}
%\boldmath
High mobility may destroy the orthogonality of subcarriers in OFDM systems, resulting in inter-carrier interference (ICI), which may greatly reduce the service quantity of high-speed railway (HSR) wireless communications. This paper focuses on ICI mitigation in the HSR downlinks with distributed transmit antennas. For such a system, its key feature is that the ICIs are caused by  multiple carrier frequency offsets corresponding to multiple transmit antennas. Meanwhile, the channel of HSR is fast time varying, which is another big challenge in the system design. In order to get a good performance, low complexity real-time ICI reduction is necessary. To this end, we first analyzed the property of the ICI matrix in AWGN and Rician scenarios, respectively. Then, we propose corresponding low complexity ICI reduction methods based on location information. For evaluating the effectiveness of the proposed method, the expectation and variance of remaining interference after ICI reduction is analyzed with respect to Rician $K$-factor. In addition, the service quantity and the bandwidth and computation cost are also discussed. Numerical results are presented to verify our theoretical analysis and the effectiveness of proposed ICI reduction methods. One important observation is that our proposed ICI mitigation method can achieve almost the same service quantity with that obtained on the case without ICI when the train's velocity is 300km/h, that is, ICI has been completely eliminated. Numerical results also show that the scenarios with Rician $K$-factors over 30dB can be considered as AWGN scenarios, which may provide valuable insights on future system designs.

\end{abstract}


\IEEEpeerreviewmaketitle

%% 引言：
\section{Introduction}\label{Sec:Introduction}

\lettrine[lines=2]{I}{n} recent decade, high-speed railway (HSR) is experiencing explosive growth and the velocity of trains can be 350km/h or higher in recent future. HSR creates special conditions and challenges for wireless channel access, due to high velocity and rapidly changing environmental conditions. For instance, the orthogonal frequency-division multiplexing (OFDM) techniques adopted by the long term evolution for railway (LTE-R) are extremely sensitive to frequency errors and can be seriously affected by carrier frequency offset (CFO), phase noise, timing offset and Doppler spread, which could trigger the inter-carrier interference (ICI) and deteriorate in high mobility scenarios \cite{li2012radio,luo2013efficient}. In addition, owing to the high mobility feature in HSR, the frequency handover between adjacent base stations (BSs) is a big challenge in system design. In \cite{yeh2010theory,wang2012distributed}, it proposed distributed antenna systems to combat the frequent handover, which introduce a new infrastructure building, increasing the signal to noise ratio (SNR) to user equipments. However, the corresponding ICI problem has not been involved. That is, ICI reduction schemes for such a HSR system need to be carefully designed.

In HSR scenarios, the train may travel in a variety of terrains such as wide plains, viaducts, mountain areas, urban and suburban districts, and tunnels. Thereby, wireless propagation environments for HSR are extremely diverse, in which viaducts and tunnels are the two typical scenarios \cite{luo2013efficient,dong2015power}. In addition to tunnels, viaducts account for the vast majority of Chinese HSRs (more specifically, 86.5\% of railways is elevated in the Beijing–Shanghai HSR \cite{dong2015power}), and there are few multi-paths because of little scattering and reflection. Hence, line-of-sight (LOS) assumption is widely used and the investigation of LOS MIMO channels are presented in \cite{dong2015power,li2013channel}.

Based on LOS assumption, it can be observed that in HSR scenarios, the distribution of arrival of angles (AOAs) is not uniform and may be discrete for the distributed antenna regimes. This is because the Doppler spread is caused by few CFOs corresponding to multiple LOS downlinks. Previous works aimed at situations where the ICI is caused by single CFO or the Doppler spread is with uniformly distributed arrival of angles (AOAs), i.e. Jakes' model\cite{jakes1994microwave}. The corresponding ICI mitigation schemes include frequency equalization \cite{schniter2004low}, ICI self-cancellation \cite{ma2012reduced}, windowing and coding\cite{zhang2003optimum}, etc. For the ICI caused by single CFO, an effective ICI reduction scheme based on the unitary property of ICI matrix is proposed in \cite{fu2005transmitter}, which may be borrowed into the cases with multiple CFOs. 

On the other hand, the wireless channels in HSR are fast time varying, which indicates low complexity real-time ICI reduction is necessary. Therefore, the effective ICI reduction methods for HSR need to be designed by taking into account the new Doppler spread characteristics and the channel fast time varying nature.

The main contributions of this paper are listed as follows:
\begin{enumerate}
\item We first analyze the ICI matrix corresponding to LOS paths. In this scenario, the ICI matrix is caused by Doppler spread with non-uniformly distributed AOAs.

\item Based on the Rician propagation model, we also derived the ICI matrix corresponding to NLOS paths, whose AOAs and amplitudes are uniformly and Gaussian distributed.

\item We prove that the overall ICI matrix is the weighting average of ICI matrices caused by LOS and NLOS paths. The weighting factors are the large-scale fading coefficients. Then, we proved that the overall ICI matrix is almost unitary.

\item Based on the theoretical results, we proposed two low complexity ICI reduction methods aiming at AWGN and Rician channels, which can avoid the matrix inverse calculation and is suitable to the fast time varying scenarios. Finally, the variances of reduced SIR with proposed ICI reduction method are derived.

\item One important observation is obtained that our proposed ICI mitigation method can achieve almost the same service quantity with that obtained on the case without ICI when the velocity of the train is 300km/h. Another observation is that Rician channels with Rician factors over 30dB can be considered as AWGN channels.
\end{enumerate}

The rest of this paper is organized as follows. Firstly, in Section \ref{Sec:SysModel}, the system model of HSR MIMO-OFDM downlinks with distributed antennas is introduced and the Doppler spread effects of LOS paths are mathematically summarized via an ICI matrix. As to the Doppler spread effects of NLOS paths and the propagation model with Rician fading, they are analyzed in section \ref{Sec:AlalysisRician}. In Section \ref{Sec:ICIReduction}, we  proposed two low complexity ICI reduction methods aiming at AWGN and Rician channels. Later on, the mean of remaining signal to interference (SIR) after ICI reduction is analyzed. In section \ref{sec:EffectivenessValuation}, we derived the variance of remaining SIR and accumulate service quantity (ASQ) with respect to Rician factor and the velocity of train. In Section \ref{Sec:NumRes}, the effectiveness of proposed method and the accuracy of previous theoretical results are verified via numerical results. Finally, the conclusions are given in Section \ref{Sec:Conclusion}.

\emph{Notation}: $(\cdot)^{-1}$, $(\cdot)^{T}$denote the inverse and transpose of $(\cdot)$, respectively. The symbols $\mathbf{E}$ and $\mathbf{G}$ denote the identical matrix and the matrix full with elements 1, respectively. $||(\cdot)||_2$ denotes the $\rm{L}_2$ norm of $(\cdot)$.

\section{System Structure and Channel Model of LOS Paths} \label{Sec:SysModel} % 系统模型，简化标准以及结论模型

%In this section, we shall review the structure of MIMO-OFDM systems with distributed antennas on HSR scenarios, in which the system model is reviewed and the effects of Doppler spread to inter carrier interference (ICI) are analyzed. Then, for the convenience of analyzing the effectiveness of proposed ICI reduction method, the Rician channel fading model with Doppler spread is proposed. To our best knowledge, the Doppler spread effects to OFDM systems with Rician fadings haven't been analyzed. Previous works mainly focused on the LOS assumption \cite{dong2015power,li2013channel} or the estimation of Doppler spread \cite{qu2016journey,gao2015subspace,bellili2013low,qu2015two}.

%\subsection{MIMO-OFDM Systems with Distributed Antennas}\label{Sec:ReviewOfDistributedMIMO-OFDM}

We consider the distributed antenna system shown in Fig. \ref{Fig:CoverageModel}, where the base band unit (BBU) is connected with a serial of radio remote units (RRUs) by optical fibers. At each transmit slots, the RRUs transmit uniform signals to receive antennas. In addition, the distributed RRUs as well as the receivers at trains are equipped with multiple antennas, which forms the multiple input multiple output (MIMO) regimes at the downlinks of HSR communication systems. The antenna number at each individual RRUs and receivers is denoted by $T_x$ and $T_y$, respectively.

%% Rail Way coverage model
\begin{figure*}[htbp]
\centering
\includegraphics[width=0.85\textwidth]{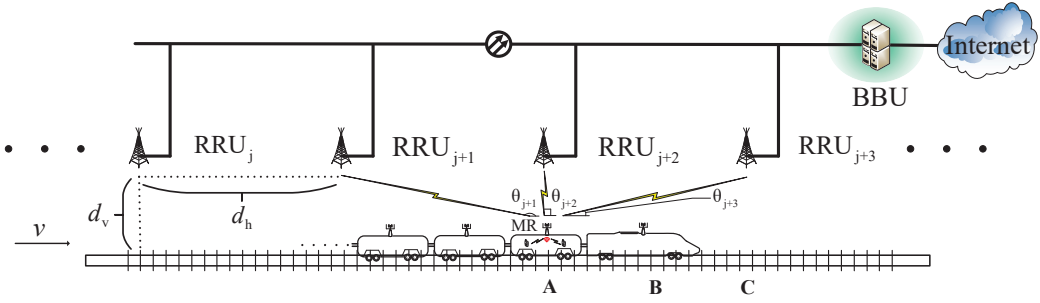}
\caption{Distributed antennas on rail way coverage system. The RRUs are connected to BBU with optical fibers and transmit uniform signals to the receivers at trains simultaneously.} \label{Fig:CoverageModel}
\end{figure*}

%Due to the widely used LOS assumption for HSR wireless channel modeling, HSR channels can be modeled as additive white Gaussian noise (AWGN) channels \cite{li2013channel,dong2015power}. 

Due to the widely used LOS assumption for HSR wireless channel modeling, we firstly consider the LOS paths only, which can be modeled as additive white Gaussian noise (AWGN) channels \cite{li2013channel,dong2015power}. Furthermore, considering the large-scale path loss, receivers can only identify few of the RRUs with relatively short distance. Denote the receive antenna number as $N_r$ and let the resolvable LOS path number corresponding to the $r$-th $(r=1,2,\cdots, N_r)$ receive antenna as $N_t$. As the train running along the railway, the signal transmitted from the $t$-th RRU $(t=1,2,\cdots, N_t)$ to the $r$-th receiver in HSR is modeled as

\begin{equation}\label{equ:HSRReceiverModel}
y_{r,t}[n] = \rho_{r,t} \mathbf{h}_{r,t}[n] x_t[n-\tau] + w_r[n],
\end{equation}
where $y_{r,t}[n]$, $x_t[n-\tau]$ and $w_r[n]$ are the time domain received signal, normalized transmit signal and the circular complex white Gaussian noise, respectively. $\tau$ is the time-delay of propagation path and less than the guard interval. Previous signals at the transmit and receive side are vectors composed by the transmit symbol and receive symbol at $T_x$ and $T_y$ antennas. $\mathbf{h}_{r,t}[n]$ denotes the unitary time domain channel fading matrix and $\rho_{r,t}$ is the power coefficient containing the power control of RRUs side and large-scale path loss between antenna pairs. By denoting the AOA of each LOS path as $\theta_{r,t}$ and involving CFO at each path, one can express the total receive signal at the $r$-th receive antenna as

\begin{equation}\label{equ:TimeDomainSysModel}
y_r[n] = \sum_{t=1} ^{N_t} e^{j\omega_D cos(\theta_{r,t}) n} \rho_{r,t} \mathbf{h}_{r,t}[n] x[n-\tau_{r,t}] + w_r[n],
\end{equation}
where $\omega_D$ is the maximum normalized Doppler frequency offset denoted by $\omega_D = \frac{f v T_s}{C}$. $v$ and $C$ is the traveling velocity of train and light, respectively. $f$ and $T_s$ denote the carrier frequency and the time domain sampling duration of OFDM signals. $\tau_{r,t}$ is the time-delay of each path and the largest $\tau_{r,t}$ is assumed to be less than guard interval. It should be noted that the time-delay of each path is ignored. This is because the distance difference among the paths are relatively short, comparing to the distance of electromagnetic wave propagation in one sampling duration $T_s$ or one can employ the signal delay launching method by using the relative position information between RRUs and receive antenna, which can be acquired via positioning systems, e.g. GPS, Beidou, etc.

After removing the cyclic prefix and performing discrete Fourier transformation on the time domain received signals in \eqref{equ:TimeDomainSysModel}, the demodulated signal at $r$-th receive antenna can be represented as \eqref{equ:DFTICIExpression} at the top of next page. In \eqref{equ:DFTICIExpression}, $Y_r[k]$, $W_r[k]$, $\mathbf{H}_{r,t}[k]$ and $X[k]$ are the frequency domain received signal, AWGN, channel fading matrix and transmitted signal at the $k$-th subcarrier ($k = 1,2,\cdots,N$), respectively. $N$ is the total subcarrier number of OFDM system and $I_{r,t}[n-k]$ is ICI coefficient between the $n$-th and $k$-th subcarriers, which can be expressed as\cite{zhao2001intercarrier,fu2005transmitter}

\begin{equation}\label{equ:I_rk[n-k]}
\begin{split}
I_{r,t}[n-k] = \frac{\rm{sin}(\pi(n+\varepsilon_{r,t}-k))}{N \rm{sin}\left( \frac{\pi}{N}(n+\varepsilon_{r,t}-k) \right)}\cdot \rm{exp}\left(j \pi \left( 1- \frac{1}{N} \right) (n+\varepsilon_{r,t}-k) \right),
\end{split}
\end{equation}
where $\varepsilon_{r,t} = \omega_D cos(\theta_{r,t})$ is the normalized frequency offset between the $r$-th receive antenna and the $t$-th transmit antenna.

\begin{figure*}[!t]
\normalsize
\begin{equation}\label{equ:DFTICIExpression}
\begin{split}
Y_r[k] &= \sum_{t=1} ^{N_t} \sum_{n=1}^{N} \rho_{r,t} I_{r,t}[n-k] \mathbf{H}_{r,t}[k] X[n] + W_r[k] \\
&= \underbrace{\sum_{t=1} ^{N_t} \rho_{r,t} I_{r,t}[0] \mathbf{H}_{r,t}[k] X[k]}_{\rm desired~signal} + \underbrace{\sum_{t=1} ^{N_t} \sum_{n=1,n\neq k}^{N} \rho_{r,t} I_{r,t}[n-k] \mathbf{H}_{r,t}[n] X[n] }_{\rm ICI~components} + W_r[k]
\end{split}
\end{equation}
%{
%\begin{equation}\label{equ:IHCKronecker}
%\begin{split}
% \mathbf{I}_{r,t} \otimes \mathbf{G} =
%\left[                 %左括号
%  \begin{array}{cccc}   %该矩阵一共4列，每一列都居中放置
%    I_{r,t}[0]\mathbf{G} & I_{r,t}[1]\mathbf{G} & \cdots	& I_{r,t}[N-1]\mathbf{G}\\  %第一行元素
%    I_{r,t}[-1]\mathbf{G} & I_{r,t}[0]\mathbf{G} & \cdots	& I_{r,t}[N-2]\mathbf{G}\\  %第二行元素
%    \vdots & \vdots & \vdots & \vdots\\  %第三行元素
%    I_{r,t}[-(N-1)]\mathbf{G} & I_{r,t}[-(N-2)]\mathbf{G} & \cdots    & I_{r,t}[0]\mathbf{G}\\  %第四行元素
%  \end{array}
%\right]\\                 %右括号
%\end{split}
%\end{equation}
\hrulefill
\vspace*{4pt}
\end{figure*}

Based on the assumption that the time-delay of each path is ignored, the MIMO channels can be considered as flat fading. \cite{shen2012channel} demonstrates that when in LOS propagation environments, the MIMO channel response matrix can be approximated as $\mathbf{G}$. As illustrated in Section \ref{Sec:Introduction}, $\mathbf{G}$ is the $N_r \times N_t$ matrix fulls of elements 1. Thus, $\mathbf{H}_{r,t}[1] = \mathbf{H}_{r,t}[2] = \cdots = \mathbf{H}_{r,t}[N] = \mathbf{G}$ and all frequency domain received signals at $N$ subcarriers can be represented in matrix form as

%% 写成输入输出信号矢量形式
\begin{equation}\label{equ:expressionMatFormFlatFading}
\begin{split}
\mathbf{Y}_r = \left\{\sum_{t=1}^{N_t} \rho_{r,t} \mathbf{I}_{r,t} \otimes \mathbf{G} \right\} \mathbf{X} + \mathbf{W}_r = \mathbf{S}_r^{\rm{L}} \mathbf{X} + \mathbf{W}_r,
\end{split}
\end{equation}
where $\mathbf{Y}_r$, $\mathbf{X}$ and $\mathbf{W}_r$ are the $N T_y \times 1$, $N T_x \times 1$ and $N T_y \times 1$ column vectors constituted by the permutation of corresponding symbol vectors in \eqref{equ:DFTICIExpression} from $1$ to $N$, respectively. The $N \times N$ matrix $\mathbf{I}_{r,t}$ is denoted by

\begin{equation}\label{equ:I_mat}
\begin{split}
& \mathbf{I}_{r,t} =\\
&\left[                 %左括号
  \begin{array}{cccc}   %该矩阵一共4列，每一列都居中放置
    I_{r,t}[0] & I_{r,t}[1] & \cdots	& I_{r,t}[N-1]\\  %第一行元素
    I_{r,t}[-1] & I_{r,t}[0] & \cdots	& I_{r,t}[N-2]\\  %第二行元素
    \vdots & \vdots & \vdots & \vdots\\  %第三行元素
    I_{r,t}[-(N-1)] & I_{r,t}[-(N-2)] & \cdots	& I_{r,t}[0]\\  %第四行元素
  \end{array}
\right].\\                 %右括号
\end{split}
\end{equation}
The non-diagonal elements of $\mathbf{I}_{r,t}$ correspond to the ICI components in \eqref{equ:DFTICIExpression}. The symbol $\otimes$ denotes the Kronecker product. $\mathbf{S}_r^{\rm{L}}$ is the sum of $N_t$ channel fading matrices weighted by $\rho_{r,t}$ and thus, it denotes the channel fading matrix corresponding to received signals from $N_t$ transmit antennas. The superscript of $\mathbf{S}_r^{\rm{L}}$ denotes that the channel fading matrix corresponds to the LOS paths.

%The non-diagonal elements of $\mathbf{I}_{r,t}$ correspond to the ICI components in \eqref{equ:DFTICIExpression}. The symbol $\otimes$ denotes the Kronecker product and the explicit expression of $\mathbf{I}_{r,t} \otimes \mathbf{G}$ is shown as \eqref{equ:IHCKronecker}. $\mathbf{S}_r^{\rm{L}}$ is the sum of $L$ channel fading matrices weighted by ICI factors $\mathbf{I}_{r,t}$ and thus, it denotes the overall channel fading matrix affected by ICI. The superscript of $\mathbf{S}_r^{\rm{L}}$ denotes that the overall channel fading matrix is obtained under the LOS assumption.

\section{Analysis on the Effects of NLOS paths} \label{Sec:AlalysisRician}

Previous works on HSR wireless communications mainly focus on the AWGN channels \cite{dong2015power,li2013channel} or the estimation of Doppler spread \cite{aboutorab2012new,bellili2013low,gao2015subspace,qu2015two}. In this part, we derived the Doppler spread model in HSR scenarios with Rician fading model. That is, with known Rician K-factor and velocity, the effects of non-line of sight (NLOS) paths is analyzed under the Rician channel fading model. It can be noted that the analyzed propagation model is one of the specific scenarios of base expansion model (BEM) \cite{giannakis1998basis}. This is because, in HSR scenarios, there exist a strong LOS path and the NLOS paths are assumed to be Gaussian distributed in this paper. 

%To our best knowledge, the Doppler spread effects to OFDM systems with Rician fadings haven't been analyzed. 

\subsection{Rician Channel Fading with Doppler Spread}\label{Sec:RicianDoppler}

%In this part, the effectiveness of LOS assumption is verified respect to varying Rician $K$ factors.

Review the Rician channel fading shown in Fig. \ref{Fig:ScatteringPath}, there exist a LOS path and a large number of reflected/scattered (i.e. NLOS) paths in one delay window. The unitary channel fading between one of the antenna pairs in \eqref{equ:HSRReceiverModel} can be modeled as \cite{tse2005fundamentals}

\begin{equation}\label{equ:Rician}
h_{r,t}[n] = \sqrt{\frac{K}{K+1}} e^{j\phi} + \sqrt{\frac{1}{K+1}} \mathcal{CN}(0,1),
\end{equation}
where the first and second term corresponds to the LOS path and NLOS paths, respectively. $\phi$ is the phase of LOS path and $\mathcal{CN}(0,1)$ is the aggression of large number of NLOS paths, which is circular complex Gaussian distributed. $K$ is the Rician factor and denotes the power ratio of the LOS path to NLOS paths.

\begin{figure}[htbp]
\centering
\includegraphics[width=0.55\textwidth]{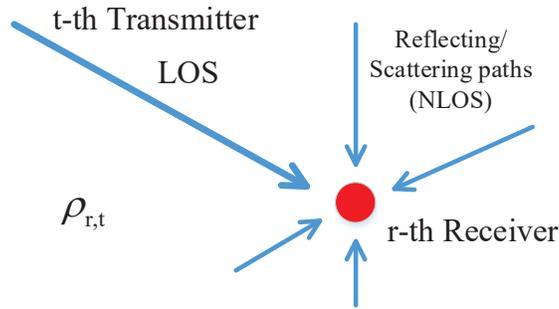}
\caption{Channel propagation model between the t-th transmitter and the r-th receiver with LOS and reflecting/scattering paths, the large-scale channel gain between antenna pair is $\rho_{r,t}$.} \label{Fig:ScatteringPath}
\end{figure}

The effect of LOS paths with Doppler shift has been analyzed in section \ref{Sec:SysModel} with \eqref{equ:I_rk[n-k]} and \eqref{equ:DFTICIExpression}. As to the NLOS paths, we assume that the amplitudes are Gaussian distributed and the path number is $N_s$ ($N_s \rightarrow \infty$). Denote the amplitudes of each paths as $a_i$ ($i=1,\cdots,N_s)$. It can be noted that, according to the central limit theorem \cite{kay2013fundamentals}, the distribution of $a_i$ doesn't affect the aggregation of NLOS paths. Thus, one can immediately derive that $a_i \sim \sqrt{\frac{1}{K+1}} \mathcal{CN}(0,\frac{1}{N_s})$. The effects of Doppler spread caused by the Doppler shifts of NLOS paths can be expressed as

\begin{equation}\label{equ:DopplerSpreadSumExpress}
D[n-k] = \sum_{i=1}^{N_s} a_i \frac{sin(\pi (n-k+\omega_D cos(\theta_i)))}{N sin(\frac{\pi}{N}(n-k + \omega_D cos(\theta_i) ))},
\end{equation}
which is the interference amplitude from the $n$-th to the $k$-th subcarrier. $\theta_i$ is the AOA of the $i$-th NLOS path and $\theta_i \sim U[-\pi,\pi]$. The statistics of $D[n-k]$ can be summarized via following lemma.

\begin{lem}\label{lem:DopplerSpreadVariance}
$D[n-k]$ is complex Gaussian distributed and

\begin{equation*}
D[n-k] \sim 
\begin{cases}
\mathcal{CN}(0,\frac{1}{K+1}) &, ~ n=k \\
\mathcal{CN}(0,\frac{\omega_D^2}{2(K+1)(n-k)^2}) &, ~ n \neq k ,\\
\end{cases}
\end{equation*}
where $\omega_D$ is the maximum normalized Doppler frequency offset.
\end{lem}
\begin{proof}
Proof of this lemma is given in appendix \ref{Appen:DopplerVariance}.
\end{proof}

It can be noticed that the variance of Doppler spread decreases linearly with the index difference $(n-k)$ and is determined by the velocity of train as well as the $K$ factor of Rician channels. In addition, the AOA of the LOS path in Fig. \ref{Fig:ScatteringPath} doesn't affect the Doppler spread of NLOS paths, which is intuitionistic.

\subsection{Propagation Model With Rician Fading}\label{Sec:PropagationWithRicianFading}
% In this part, the propagation model with Rician fading is derived. 

By substituting the results in Lemma \ref{lem:DopplerSpreadVariance} into \eqref{equ:expressionMatFormFlatFading}, the channel propagation model with Rician fading can be modified as \eqref{equ:ModifiedExpressionMatFormFlatFading}, shown at the top of next page. $\mathbf{S}_r$ is the channel fading matrix of Rician channels with Doppler spread, wherein $\mathbf{S}_r^{\rm{L}}$ and $\mathbf{S}_r^{\rm{N}}$ corresponds to the LOS and NLOS paths. $\mathbf{R}_t$ is the aggression of these paths and $\mathbf{R}_t=r_t \mathbf{G}$ ($r_t\sim \mathcal{CN}(0,1)$), when beam-forming is adopted. $\mathbf{I}_D$ is the ICI coefficient matrix caused by the Doppler shifts of NLOS paths and can be easily obtained by replacing $I_{r,t}[n-k]$ in \eqref{equ:I_mat} with $I_{D}[n-k]$, i.e.

\begin{figure*}[!t]
\normalsize
\begin{equation}\label{equ:ModifiedExpressionMatFormFlatFading}
\begin{split}
\mathbf{Y}_r = \sqrt{\frac{K}{K+1}}\left\{\sum_{t=1}^{N_t} \rho_{r,t} \mathbf{I}_{r,t} \otimes \mathbf{G} \right\} \mathbf{X} + \sqrt{\frac{1}{K+1}}\left\{\sum_{t=1}^{N_t} \rho_{r,t} \mathbf{I}_{D} \otimes \mathbf{R}_t \right\} \mathbf{X} + \mathbf{W}_r = \underbrace{(\mathbf{S}_r^{\rm{L}} + \mathbf{S}_r^{\rm{N}} )}_{\mathbf{S}_r} \mathbf{X} + \mathbf{W}_r.
\end{split}
\end{equation}
\hrulefill
\vspace*{4pt}
\end{figure*}

\begin{equation}\label{equ:I_Dmat}
\begin{split}
& \mathbf{I}_{D} =\\
&\left[                 %左括号
  \begin{array}{cccc}   %该矩阵一共4列，每一列都居中放置
    I_{D}[0] & I_{D}[1] & \cdots	& I_{D}[N-1]\\  %第一行元素
    I_{D}[-1] & I_{D}[0] & \cdots	& I_{D}[N-2]\\  %第二行元素
    \vdots & \vdots & \vdots & \vdots\\  %第三行元素
    I_{D}[-(N-1)] & I_{D}[-(N-2)] & \cdots	& I_{D}[0]\\  %第四行元素
  \end{array}
\right],\\                 %右括号
\end{split}
\end{equation}
where

\begin{equation}
I_{D}[n-k] = 
\begin{cases}
1 &, ~ n=k \\
\frac{(-1)^{n-k}\omega_D}{\sqrt{2}(n-k)} &, ~ n \neq k ,\\
\end{cases}
\end{equation}

The estimation of $\mathbf{S}_r^{\rm{L}}$ in \eqref{equ:expressionMatFormFlatFading} is actually equivalent to the estimation of large scale channel fading $\rho_{r,t}$, which can be acquired with known relative positions between RRUs and receive antennas \cite{medeisis2000use}. In addition, the estimation of $\mathbf{S}_r^{\rm{N}}$ in \eqref{equ:ModifiedExpressionMatFormFlatFading} is equivalent to the estimation of $\mathbf{R}_t$, which can be acquired via channel estimation methods. Hence, we assume that $\mathbf{S}_r$ has been perfectly estimated. In order to mitigate the ICI influence on the OFDM symbols detection, it is necessary to eliminate $\mathbf{S}_r$ in fact.

\section{ICI Reduction Schemes}\label{Sec:ICIReduction}

%This section discuss the reduction of ICI components. 

As previously illustrated, ICI components in \eqref{equ:DFTICIExpression} can be reduced by counteract operation method via repetition coding (i.e. ICI self-cancellation)\cite{ma2012reduced}, which could reduce the bandwidth efficiency for fast time varying channels. Likewise, ICI can be reduced via frequency equalization\cite{hwang2009ofdm,schniter2004low}, which can mitigate ICI efficiently but is limited by the complexity of matrix inverse calculation, especially with large amount of subcarriers. Both of the methods can either reduce spectrum efficiency or with very high computational complexity, which limit their applications in HSR scenarios.

In this section, we shall utilize the approximately unitary property of the channel fading matrix $\mathbf{S}_r$ to retain the spectrum efficiency and achieve low complexity frequency equalization for real-time ICI reduction over fast time varying HSR channels. For the scenarios with and without NLOS reflecting/scattering paths, two distinct ICI reduction methods are proposed and the effectiveness of these two methods shall be analyzed in section \ref{sec:EffectivenessValuation}.

\subsection{Orthogonal Property, ICI Reduction Method and SIR Analysis in AWGN Scenarios}\label{Sec:OrthogonalEqualizationForLOS}
Let us first analyze the properties of channel fading matrix and summarize the unitary property of $\mathbf{S}_r^{\rm{L}}$ as following lemma.

\begin{lem}\label{lem:OriginalOrthogonalLOS}
$\mathbf{S}_r^{\rm{L}}$ defined in \eqref{equ:expressionMatFormFlatFading} is approximately orthogonal, that is,
\begin{equation}\label{equ:OriginalOrthogonal}
{(\mathbf{S}_r^{\rm{L}})}^{T} \mathbf{S}_r^{\rm{L}} \approx \mathbf{E} \otimes (\bm{\beta} ^T \bm{\beta}),
\end{equation}
where $\bm{\beta} = \sum_{t=1}^{N_t} \rho_{r,t} \mathbf{G}$ is the weighted sum of LOS fading channels from $L$ RRUs.
\end{lem}
\begin{proof}
The proof of this lemma is given in appendix \ref{Appen:LOSUnitary}. 
\end{proof}

Based on the result of Lemma 1, the approximately frequency equalization can be executed by replacing $\gamma_{\beta}(\mathbf{S}_r^{\rm{L}})^{-1}$ with $(\mathbf{S}_r^{\rm{L}})^{T} / \gamma_{\beta}$. $\gamma_{\beta}^2 = {\rm{trace}} \left(({\mathbf{S}_r^{\rm{L}})}^{T} \mathbf{S}_r^{\rm{L}} \right) = N T_y T_x (\sum_{t=1}^{N_t} \rho_{r,t})^2$, so that $\gamma_{\beta}$ is the channel gain of summed fading channels from $N_t$ transmit antennas to the $r$-th receive antenna. Hence, $(\mathbf{S}_r^{\rm{L}})^T / \gamma_{\beta}$ is normalized and retains the power of frequency equalized signals $\hat{\mathbf{Y}}_r$, which can be immediately expressed as

\begin{equation} \label{equ:FrequencyEqualization}
\begin{split}
\hat{\mathbf{Y}}_r = \frac{1}{\gamma_{\beta}} (\mathbf{S}_r^{\rm{L}})^T \mathbf{S}_r^{\rm{L}} \mathbf{X} + \frac{1}{\gamma_{\beta}} (\mathbf{S}_r^{\rm{L}})^T \mathbf{W}_r = \frac{\gamma_{\beta}}{N T_x} (\mathbf{E} \otimes \mathbf{G}') \mathbf{X} + \sqrt{\gamma_{\beta}^2 / {\rm{SIR}}_k + 1}\mathbf{W}_{r}^{'},
\end{split}
\end{equation}
which avoids the complexity of calculating $\mathbf{S}_{r}^{-1}$. $\rm{SIR}_k$ denotes the SIR at the $k$-th subcarrier after approximately frequency equalization and $\mathbf{W}_{r}^{'}$ is the normalized AWGN and the noise power coefficient includes the remaining interference power $\gamma_{\beta}^2 / {\rm{SIR}}_k$ and the original normalized Gaussian noise power in \eqref{equ:expressionMatFormFlatFading}. $\mathbf{G}'$ is the $T_x \times T_x$ matrix full of elements 1. It can be observed that the ICI caused by high mobility has been eliminated.

To analyze the ${\rm{SIR}}_k$ in \eqref{equ:FrequencyEqualization}, let ${(\mathbf{S}_r^{\rm{L}})}^T \mathbf{S}_r^{\rm{L}}=\mathbf{\Lambda}_r^{\rm{L}}$. Notice the definition of Kronecker product, $\mathbf{S}_r^{\rm{L}}$ can be considered as partitioned matrix and the size of submatrices is $T_y \times T_x$. In addition, because of the frequency domain transmitted signals $X[k]$ are normalized and independent with each other, $\rm{SIR}_k$ can be expressed as

\begin{equation}\label{equ:SIRExpression}
{\rm{SIR}}_k = \frac{||\mathbf{\Lambda}_r^{\rm{L}}(k,k)||_{2}^{2}}{\sum_{i=1,\neq k}^{N} ||\mathbf{\Lambda}_r^{\rm{L}}(k,i)||_{2}^{2}},
\end{equation}
where $\mathbf{\Lambda}^{\rm{L}}_r(k,i)$ is the submatrix at the $k$-th row and the $i$-th column. The diagonal and non-diagonal submatrix $\mathbf{\Lambda}^{\rm{L}}_r(k,k)$ and $\mathbf{\Lambda}^{\rm{L}}_r(k,k+m)$ can be calculated by \eqref{equ:ResDiagonalOfL} and \eqref{equ:ResNonDiagonalOfL}, which are shown at the top of next page. Hence, the $\rm{SIR}_k$ can be derived by substituting \eqref{equ:ResDiagonalOfL} and \eqref{equ:ResNonDiagonalOfL} into \eqref{equ:SIRExpression}. 

%\newcounter{mytempeqncnt}
\begin{figure*}[!t]
\normalsize
\begin{equation} \label{equ:ResDiagonalOfL}
\mathbf{\Lambda}_r^{\rm{L}}(k,k) = \sum_{i=1}^{N_t} \sum_{j=1}^{N_t} \frac{\rho_{r,i} \rho_{r,j} sin(\pi \varepsilon_{r,i}) sin(\pi \varepsilon_{r,j})}{\pi^2 \varepsilon_{r,i} \varepsilon_{r,j}} \mathbf{G}^T \mathbf{G} + \sum_{i=1}^{N_t} \sum_{j=1}^{N_t} \frac{\rho_{r,i} \rho_{r,j} sin(\pi \varepsilon_{r,i}) sin(\pi \varepsilon_{r,j})}{3} \mathbf{G}^T \mathbf{G}
\end{equation}

\begin{equation} \label{equ:ResNonDiagonalOfL}
\begin{split}
\mathbf{\Lambda}_r^{\rm{L}}(k,k+m) = &\bigg\{\sum_{i=1}^{N_t} \sum_{j=1}^{N_t} \frac{\rho_{r,i} \rho_{r,j} sin(\pi \varepsilon_{r,i}) sin(\pi \varepsilon_{r,j})}{\pi^2 \varepsilon_{r,i}(\varepsilon_{r,j} - m)}\mathbf{G}^T \mathbf{G} + \sum_{i=1}^{N_t} \sum_{j=1}^{N_t} \frac{\rho_{r,i} \rho_{r,j} sin(\pi \varepsilon_{r,i}) sin(\pi \varepsilon_{r,j})}{\pi^2 (\varepsilon_{r,i} + m)\varepsilon_{r,j}}\mathbf{G}^T \mathbf{G} \\
&+ 2 \sum_{i=1}^{N_t} \sum_{j=1}^{N_t} \frac{\rho_{r,i} \rho_{r,j} sin(\pi \varepsilon_{r,i}) sin(\pi \varepsilon_{r,j})}{\pi^2 m^2}\mathbf{G}^T \mathbf{G} \bigg\} (-1)^m
\end{split}
\end{equation}
\hrulefill
\vspace*{4pt}
\end{figure*}

It can be noticed that the precise expression of ${\rm{SIR}}_k$ is much too complex and can not provide brief insight on the performance of proposed ICI reduction method. Thus, a closed form approximation is provided in section \ref{sec:EffectivenessValuation} by the maximum and minimal ${\rm{SIR}}_k$ and the accuracy has been verified by simulation results. The SIR dynamic range at the $k$-th subcarrier as the train running alone the railway can be simplified as

\begin{equation}\label{equ:SIRApprox}
\begin{cases}
{\rm{max}}({\rm{SIR}}_k) &= \frac{\psi}{72 \omega_D^{4} cos^4(\theta_A) \zeta(4)}\\
{\rm{min}}({\rm{SIR}}_k) &= \frac{1}{8 \omega_{D}^4 cos^4(\theta_B) \zeta(4)}.
\end{cases}
\end{equation}
The variables $\psi = \frac{\rho_{j+2}^2}{\rho_{j+1}^2}$ and $\cos(\theta_A) = d_{\rm{h}} / \sqrt{d^2_{\rm{h}} + d^2_{\rm{v}}}$ are the ratio of $\rho_{j+2}^2$ to $\rho_{j+1}^2$ and the cosine of signal AOA from the $(j+3)$-th RRU, when the train is at A in Fig. \ref{Fig:CoverageModel}. $\zeta(4)=\frac{\pi^4}{90}$ is the Riemann Zeta function\cite{titchmarsh1986theory}. Proofs of \eqref{equ:ResDiagonalOfL}, \eqref{equ:ResNonDiagonalOfL}, \eqref{equ:SIRApprox} are also given in appendix \ref{Appen:LOSUnitary}.

The orthogonality based frequency equalization scheme can be summarized as algorithm \ref{alg:ICIReductionAlgLOS}. It can be observed that due to $\mathbf{S}_r$ can be expressed via the weighted sum of $\mathbf{I}_{r,t} \otimes \mathbf{G}$, proposed algorithm avoids the computation complexity and bandwidth cost for Doppler spread estimation, i.e. the overall channel fading matrix affected by Doppler spread can be easily generated via the known train velocity and relative positions between RRUs and receive antennas. This is because the fast time-varying small-scale channel fading information is omitted and matrix $\mathbf{G}$ is fixed and known, while the large-scale channel fading can be predicted via specific methods, e.g. the Okumura-Hata propagation prediction model \cite{medeisis2000use}. In addition, employing $(\mathbf{S}_r^{\rm{L}})^T$ rather than $(\mathbf{S}_r^{\rm{L}})^{-1}$ reduces the computational complexity of frequency equalization.

\begin{algorithm}[htb] %算法的开始
\caption{Frequency Equalization in LOS Scenarios} %标题
\label{alg:ICIReductionAlgLOS} %标记算法，方便在其它地方引用
\begin{algorithmic}

\REQUIRE ~~ 
\begin{enumerate}
\item The set of system parameters, \{$T_x,T_y, N$\};
\item Received signals, $\mathbf{Y}_r$;  
%\item Training sequence on pilots, $T_s$;  
\item Map of large-scale fading, $\mathbf{M}_L$;  
\item Velocity of trains, $v$.
\end{enumerate}

\ENSURE Frequency equalized signals, $\hat{\mathbf{Y}}_r$ 

\textbf{Initialization:} Search large-scale fading infos (i.e. $\rho_{r,t}$) from $\mathbf{M}_L$.

\textbf{While} large-scale fading infos remains stable \textbf{do}

\begin{enumerate}
\item Generate $\mathbf{I}_{r,t}$ with \eqref{equ:I_mat};

\item $\mathbf{S}_r^{\rm{L}}$ $\leftarrow$ $\sum_{t=1}^L \rho_{r,t} \mathbf{I}_{r,t} \otimes \mathbf{G}$.

\item $\gamma_{\beta}$ $\leftarrow$ $\sqrt{N T_y T_x}(\sum_{t=1}^L \rho_{r,t})$;

\item $\hat{\mathbf{Y}}_r$ $\leftarrow$ $\frac{1}{\gamma_{\beta}} \left(\mathbf{S}_r^{\rm{L}}\right)^T \mathbf{Y}_r$;

\end{enumerate}

\textbf{Return:} $\hat{\mathbf{Y}}_r$

\end{algorithmic}
\end{algorithm}

\subsection{Orthogonal Property, ICI Reduction Method and SIR Analysis in Rician Fading Scenarios}\label{Sec:OrthogonalEqualizationForRician}

For the scenarios with Rician fading channels, more channel information and computation complexity is need to to reduce ICI. This is because the random part in \eqref{equ:Rician}, which is caused by the random scattering/reflecting, needs to be estimated and eliminated. Similarly, the orthogonal property holds for the channel fading matrices with Rician Doppler propagation model, which can be expressed as following Lemma.

\begin{lem}\label{lem:OriginalOrthogonalRician}
$\mathbf{S}_r$ defined in \eqref{equ:ModifiedExpressionMatFormFlatFading} is approximately orthogonal, that is,
\begin{equation}\label{equ:OriginalOrthogonal}
\mathbf{S}_r^{T} \mathbf{S}_r \approx \mathbf{E} \otimes (\bm{\zeta} ^T \bm{\zeta}),
\end{equation}
where $\bm{\zeta} = \sum_{t=1}^{N_t} \rho_{r,t} \left( \sqrt{\frac{K}{K+1}} \mathbf{G}+ \sqrt{\frac{1}{K+1}}\mathbf{R}_t \right) $ is the weighted sum of Rician fading channels from $L$ RRUs.
\end{lem}
\begin{proof}
The proof of this lemma is given in appendix \ref{Appen:RicianUnitary}. 
\end{proof}

Thus, the approximately frequency equalization can also be executed by replacing $\gamma_{\zeta} \mathbf{S}_r^{-1}$ with $\mathbf{S}_r^T / \gamma_{\zeta}$, where $\gamma_{\zeta}^2 = \rm{trace}(\bm{\beta}^T \bm{\beta})$ and $\gamma_{\zeta}$ is the channel gain of summed fading channels from $N_t$ transmit antennas to the $r$-th receive antenna. Hence, $\mathbf{S}_r^T / \gamma_{\zeta}$ is normalized and unitary. The frequency equalized signals $\hat{\mathbf{Y}}_r$ can be immediately expressed as

\begin{equation} \label{equ:FrequencyEqualization}
\begin{split}
\hat{\mathbf{Y}}_r &= \frac{1}{\gamma_{\zeta}} \mathbf{S}_r^T \mathbf{S}_r \mathbf{X} + \frac{1}{\gamma_{\zeta}} \mathbf{S}_r^T \mathbf{W}_r = \gamma_{\zeta} \mathbf{X} + \sqrt{\gamma_{\zeta}^2 / \rm{SIR}_k + 1}\mathbf{W}_{r}^{'},
\end{split}
\end{equation}
which avoids the complexity of calculating $\mathbf{S}_{r}^{-1}$ but involves the estimation of small-scale channel fading $\mathbf{R}_t$ compared with the ICI reduction method proposed in Section \ref{Sec:OrthogonalEqualizationForLOS}. ${\rm{SIR}}_k$ and $\mathbf{W}_{r}^{'}$ has the same definition with previous illustrated.

To analyze the ${\rm{SIR}}_k$ in \eqref{equ:FrequencyEqualization}, let $\mathbf{S}_r^T \mathbf{S}_r=\mathbf{\Lambda}_r$, $\left(\mathbf{S}_r^{\rm{N}}\right)^T \mathbf{S}_r^{\rm{N}}=\mathbf{\Lambda}_r^{\rm{N}}$, $\left(\mathbf{S}_r^{\rm{L}}\right)^T \mathbf{S}_r^{\rm{N}}=\mathbf{\Lambda}_r^{\rm{LN}}$ and $\left(\mathbf{S}_r^{\rm{N}}\right)^T \mathbf{S}_r^{\rm{L}}=\mathbf{\Lambda}_r^{\rm{NL}}$, which are partitioned matrices and the size of submatrices is $T_y \times T_x$. The diagonal and non-diagonal submatrices of $\mathbf{\Lambda}_r^{\rm{N}}$ is

\begin{equation} \label{equ:ResDiagonalOfN}
\begin{split}
\mathbf{\Lambda}_{r}^{\rm{N}}(k,k) = \left(1 + \frac{\pi^2 \omega_D^2}{6} \right)\sum_{i=1}^{L} \sum_{j=1}^{L} \rho_{r,i} \rho_{r,j} \mathbf{R}_i^T \mathbf{R}_j 
\end{split}
\end{equation}
and
\begin{equation} \label{equ:ResNonDiagonalOfN}
\begin{split}
\mathbf{\Lambda}_{r}^{\rm{N}}(k,k+m) = &\bigg\{\sum_{i=1}^{L} \sum_{j=1}^{L} \frac{\rho_{r,i} \rho_{r,j} \omega_D^2}{m^2}\mathbf{R}_{i}^T \mathbf{R}_{j} \bigg\}(-1)^m.
\end{split}
\end{equation}
Also, $\mathbf{\Lambda}_r^{\rm{LN}}$ can be expressed as \eqref{equ:ResDiagonalOfLN} and \eqref{equ:ResNonDiagonalOfLN}, which are shown at the top of next page. The expression of $\mathbf{\Lambda}_r^{\rm{NL}}$ is omitted, which can be easily obtained by $(\mathbf{\Lambda}_r^{\rm{LN}})^T$. Referencing the expression in \eqref{equ:SIRExpression}, the SIR at the $k$-th subcarrier can be immediately expressed as

%\newcounter{mytempeqncnt}
\begin{figure*}[!t]
\normalsize

\begin{equation} \label{equ:ResDiagonalOfLN}
\mathbf{\Lambda}_{r}^{\rm{LN}}(k,k) = \sum_{i=1}^{N_t} \sum_{j=1}^{N_t} \frac{\rho_{r,i} \rho_{r,j} sin(\pi \varepsilon_{r,i})}{\pi \varepsilon_{r,i}} \mathbf{G}^T \mathbf{R}_{j} + \frac{\sqrt{2}\pi}{6}\sum_{i=1}^{N_t} \sum_{j=1}^{N_t} \rho_{r,i} \rho_{r,j} sin(\pi \varepsilon_{r,i}) \omega_D \mathbf{G}^T \mathbf{R}_{j}
\end{equation}

\begin{equation} \label{equ:ResNonDiagonalOfLN}
\begin{split}
\mathbf{\Lambda}_{r}^{\rm{LN}}(k,k+m) = \bigg\{ \sum_{i=1}^{N_t} \sum_{j=1}^{N_t} \frac{\rho_{r,i} \rho_{r,j} \omega_D}{-\sqrt{2}m}\mathbf{G}^T \mathbf{R}_{j} &+ \sum_{i=1}^{N_t} \sum_{j=1}^{N_t} \frac{\rho_{r,i} \rho_{r,j} \sin(\pi\varepsilon_{r,i})}{\pi(\varepsilon_{r,i} + m)}\mathbf{G}^T \mathbf{R}_{j}\\
&+ \sqrt{2} \sum_{i=1}^{N_t} \sum_{j=1}^{N_t} \frac{\rho_{r,i} \rho_{r,j} \varepsilon_{r,i} \omega_D}{m^2}\mathbf{G}^T \mathbf{R}_{j} \bigg\} (-1)^m
\end{split}
\end{equation}

\begin{equation}\label{equ:SIRApproxRicianMax}
\begin{split}
&{\rm{max(SIR}}_k{\rm{)}}\\
& =  \frac{||K \rho_{j+2}^2 \mathbf{G}^T \mathbf{G} + 2\sqrt{K}\rho_{j+2}^2\mathbf{G}^T \mathbf{R}_{j+2} + \rho_{j+2}^2 \mathbf{R}_{j+2}^T \mathbf{R}_{j+2}||_2^2}{2\zeta(4) ||6K \rho_{j+1}\rho_{j+2}  {\omega_D}^2 \cos^2(\theta_A) \mathbf{G}^T \mathbf{G} - 4\sqrt{K}\rho_{j+1}\rho_{j+2}\omega_D^2 \cos^2(\theta_A) \mathbf{G}^T \mathbf{R}_{j+2} + 2 \rho_{j+2}^2 \omega_D^2 \mathbf{R}_{j+2}^T \mathbf{R}_{j+2} ||_2^2}
\end{split}
\end{equation}

\begin{equation}\label{equ:SIRApproxRicianMin}
\begin{split}
{\rm{min(SIR}}_k{\rm{)}} = \frac{||4K\rho_{j+2}^2 \mathbf{G}^T \mathbf{G} + 2 \sqrt{K} \rho_{j+2}^2 \mathbf{G}^T(\mathbf{R}_{j+1} + \mathbf{R}_{j+2}) + \rho_{j+2}^2 (\mathbf{R}_{j+1}+\mathbf{R}_{j+2})^T (\mathbf{R}_{j+1}+\mathbf{R}_{j+2}) ||_2^2}{2\zeta(4) ||8K \rho_{j+2}^2 \omega_D^2 cos^2(\theta_B) \mathbf{G}^T \mathbf{G} + \rho_{j+2}^2 \omega_D^2 (\mathbf{R}_{j+1}+\mathbf{R}_{j+2})^T (\mathbf{R}_{j+1}+\mathbf{R}_{j+2})||_2^2}
\end{split}
\end{equation}
\hrulefill
\vspace*{4pt}
\end{figure*}

\begin{equation}\label{equ:SIRRicianExpression}
{\rm{SIR}}_k = \frac{|| \mathbf{\Lambda}_r(k,k) ||_2^2}{\sum_{i=1,\neq k}^{N} || \mathbf{\Lambda}_r(k,i) ||_{2}^{2}},
\end{equation}
where 
\begin{equation}
\begin{split}
\mathbf{\Lambda}_r(k,i) = \frac{K}{K+1} & \mathbf{\Lambda}_r^{\rm{L}}(k,i) + \frac{\sqrt{K}}{K+1}\mathbf{\Lambda}_r^{\rm{LN}}(k,i) + \frac{\sqrt{K}}{K+1}\mathbf{\Lambda}_r^{\rm{NL}}(k,i) + \frac{1}{K+1}\mathbf{\Lambda}_r^{\rm{N}}(k,i).
\end{split}
\end{equation}

The corresponding maximum and minimal $\rm{SIR}_k$ can be simplified as \eqref{equ:SIRApproxRicianMax} and \eqref{equ:SIRApproxRicianMin}, which is shown at the top of next page and is the SIR dynamic range at the $k$-th subcarrier as the train running alone the railway. The mean of \eqref{equ:SIRApproxRicianMax} and \eqref{equ:SIRApproxRicianMin} can also be expressed as \eqref{equ:SIRApprox}. The derivation of the mean $\rm{SIR}_k$ and \eqref{equ:ResDiagonalOfN}-\eqref{equ:SIRRicianExpression} are shown in appendix \ref{Appen:RicianUnitary}.

Based on previous theoretical results, the orthogonality based frequency equalization scheme can be summarized as Algorithm \ref{alg:originalAlgRician}.

\begin{algorithm}[htb] %算法的开始

\caption{Frequency Equalization in Rician Scenarios} %标题
\label{alg:originalAlgRician}%标记算法，方便在其它地方引用
\begin{algorithmic}

\REQUIRE ~~ 
\begin{enumerate}
\item The set of system parameters, \{$T_x,T_y, N$\};
\item Received signals, $\mathbf{Y}_r$;  
\item Map of large-scale fading, $\mathbf{M}_L$;
\item Training sequence on pilots, $T_s$;  
\item Large-scale fading infos, $\rho_{r,t}$ ;  
\item Velocity of trains, $v$.
\end{enumerate}

\ENSURE Frequency equalized signals, $\hat{\mathbf{Y}}_r$ 

\textbf{Initialization:} Search large-scale fading infos (i.e. $\rho_{r,t}$) from $\mathbf{M}_L$ and estimate the channel fading $\bigg\{ \sqrt{\frac{K}{K+1}} \mathbf{G}+ \sqrt{\frac{1}{K+1}}\mathbf{R}_t \bigg\}$ with $T_s$.

% generate $\mathbf{S}_r$, $\bm{\beta}$ and $\gamma$ with Lemma \ref{lem:OriginalOrthogonal}.
\textbf{While} the large- and small-scale channel fading remains stable \textbf{do}

\begin{enumerate}
\item Generate $\mathbf{I}_{r,t}$ with \eqref{equ:I_mat};

\item $\mathbf{S}_r$ $\leftarrow$ $\mathbf{S}_{r}^{L} + \mathbf{S}_{r}^{N}$, where $\mathbf{S}_{r}^{L}$ and $\mathbf{S}_{r}^{N}$ are generated from \eqref{equ:ModifiedExpressionMatFormFlatFading};

\item $\gamma_{\zeta}$ $\leftarrow$ ${\rm{trace}} \left\{\sum_{t=1}^{N_t} \rho_{r,t} \left( \sqrt{\frac{K}{K+1}} \mathbf{G}+ \sqrt{\frac{1}{K+1}}\mathbf{R}_t \right) \right\}$; 

\item $\hat{\mathbf{Y}}_r$ $\leftarrow$ $\frac{1}{\gamma_{\zeta}} \mathbf{S}_r^T \mathbf{Y}_r$;

\end{enumerate}

\textbf{Return:} $\hat{\mathbf{Y}}_r$

\end{algorithmic}
\end{algorithm}

\section{Evaluation of Proposed ICI Reduction Methods}\label{sec:EffectivenessValuation}

In this section, let us consider the situation when the power of NLOS paths can not to be neglected. The effectiveness of proposed ICI reduction is considered respect to varying Rician $K$ factors. Later, we shall analyze the impact of remaining SIR after ICI reduction from the perspective of information theory and evaluate how mobility affects the ASQ of proposed ICI reduction method. ASQ \cite{dong2012deterministic} is the accumulate channel service capacity in one period. In addition, the approximate calculation of ASQ is proposed, which gives more specific instructions for the design of HSR wireless communication systems. Finally, the computation complexity and bandwidth cost of proposed ICI reduction methods is analyzed.

\subsection{Effectiveness with Respect to K Factor}

In proposed ICI reduction methods, the frequency equalization is carried out by applying $(\mathbf{S}_r^L)^T / \gamma_{\beta}$ or $\mathbf{S}_r^T / \gamma_{\zeta}$ to both sides of \eqref{equ:ModifiedExpressionMatFormFlatFading} in AWGN or Rician scenarios, respectively. Thus, following two questions should be answered: 1) When the propagation channels can be considered as AWGN channels. 2) How the NLOS paths affects ICI reduction results. That is, the suitable situation of proposed ICI reduction methods and the effects of varying K factor on AWGN or Rician scenarios need to be analyzed. For the concise of paper, the theoretical results are summarized in following lemma.

\begin{lem}\label{lem:VarianceRespectToRiceFactor}
In ICI reduction method aiming at AWGN scenarios, the expectation and variance of maximum $\rm{SIR}_k$ with respect to Rician K-factor can be expressed as

\begin{equation}\label{equ:AveAndVarOfSimMaxSIR}
\begin{cases}
\mathbb{E}\left({\rm{max}} \left(\rm{SIR}_k\right)\right) &= \frac{\psi K}{72\omega_D^4\cos^4(\theta_A)\zeta(4) K + \psi \zeta(2)\omega_D^2}\\
\mathrm{Var}\left({\rm{max}} \left(\rm{SIR}_k\right)\right) &= 8 \log_{10}^2\left(1 + \frac{\psi \zeta(2)}{29\omega_D^2 \cos^4(\theta_A) \zeta(4)K} \right) {\rm{dB}}^2.\\
\end{cases}
\end{equation}
In ICI reduction method aiming at Rician scenarios, the variance of $\rm{SIR}_k$ with respect to Rician K-factor can be expressed as

\begin{equation}\label{equ:VarOfOriginalSIR}
\begin{cases}
\mathrm{Var}\left({\rm{max}} \left(\rm{SIR}_k\right)\right) &= \frac{400}{\ln^2(10) K} ~{\rm{dB}}^2\\
\mathrm{Var}\left({\rm{min}} \left(\rm{SIR}_k\right)\right) &= \frac{200}{\ln^2(10) K} ~{\rm{dB}}^2.\\
\end{cases}
\end{equation}

\end{lem}
\begin{proof}
Proof of this lemma is given in appendix \ref{Appen:Variance}. 
\end{proof}

It can be observed from Lemma \ref{lem:VarianceRespectToRiceFactor} that, in ICI reduction method aiming at AWGN scenarios, the expectation of maximum ${\rm{SIR}_k}$ increases with respect to Rician K-factor. Meanwhile, the variance of ${\rm{SIR}_k}$ decreases with respect to K. Thus, in scenarios with properly high K-factor, the effects of NLOS parts can be neglected. Using the variance to expectation ratio (VER), we define the proper K-factor as follows.

\begin{defn}\label{def:ProperlyK}
Define the Rician K-factor $K_p$ satisfies that $\frac{\mathrm{Var}\left({\rm{max}} \left(\rm{SIR}_k\right)\right)}{\mathbb{E}\left({\rm{max}} \left(\rm{SIR}_k\right)\right)} = \chi$, where $\chi$ is the preset threshold (i.e. 10\%).
\end{defn}

%It can be observed that $\mathrm{Var}\left({\rm{max}} \left(\rm{SIR}_k\right)\right)$ decreases with respect to K, while $\mathbb{E}\left({\rm{max}} \left(\rm{SIR}_k\right)\right)$ increases. Thus, f

For channels with $K \geq K_p$, adopting Alg. \ref{alg:ICIReductionAlgLOS} can acquire less VER than $\chi$. In these scenarios, propagation channels can be considered as AWGN channels. On the contrary, for channels with $K \leq K_p$, Alg. \ref{alg:originalAlgRician} should be adopted to provide better ICI reduction performance, including higher expectation and less VER of remaining ${\rm{SIR}_k}$. In these scenarios, propagation channels should be considered as Rician channels. The proof can be easily derived with the monotony of $\mathbb{E}\left({\rm{max}} \left(\rm{SIR}_k\right)\right)$ and $\mathrm{Var}\left({\rm{max}} \left(\rm{SIR}_k\right)\right)$.

\subsection{Accumulate Service Quantity}

Without loss of generality, let A as the original, the position $x$ from $A \rightarrow C$ varies from $0$ to $d_h$, which is the distance between adjacent RRUs. In addition, denote the corresponding signal to interference plus noise ratio at the $k$-th subcarrier ($\rm{SINR}_k$) as ${\rm{SINR}_k}(x)$. Hence, the ASQ $\mathcal{L}(\omega_{D},x)$ of $k$-th subcarrier can be represented as\cite{dong2012deterministic}

\begin{equation}\label{equ:ApproxMobileService}
\begin{split}
\mathcal{L}(\omega_{D},&\psi) =  \int_0^{x/v} \log_2 \left(1 + \frac{\gamma^2}{ \gamma^2 / {\rm{SINR}_k}(\tau v) + 1}\right) d\tau\\
& \thickapprox \frac{x}{2v} \bigg \{ \log_2 \left(1 + \frac{\gamma^2}{\gamma^2 / \mathbb{E}\big(\rm{max}(\rm{SINR}_k)\big) + 1}\right) + \log_2 \left(1 + \frac{\gamma^2}{\gamma^2/ \mathbb{E}\big(\rm{min}(\rm{SINR}_k)\big) + 1}\right) \bigg \}.
\end{split}
\end{equation}
The approximation of \eqref{equ:ApproxMobileService} is by using the the linear approximation of ${\rm{SINR}_k}(x)$ with respect to $x$, which will be verified by simulations in next Section. If $x=\psi$, the ASQ is the mobile service (MS) and denotes the accumulated channel capacity in one period. Observing \eqref{equ:SIRApprox} and \eqref{equ:ApproxMobileService}, one can find that as the mobility increases, $\omega_{D}$ increases as well and hence, reduces the channel mobile service $\mathcal{L}(\omega_{D},\psi)$. On the other hand, one can conclude that the MS at one deterministic subcarrier is only related to the velocity when the K-factor as well as the distributions and transmit power of RRUs are determined.

\subsection{Computational Complexity and Bandwidth Cost}

In AWGN and Rician scenarios, Alg. \ref{alg:ICIReductionAlgLOS} and Alg. \ref{alg:originalAlgRician} is adopted, respectively. The difference of two proposed ICI reduction method is whether or not estimating the small-scale channel fading. In literature, there are amount of channel estimation methods. For instance, in \cite{van1995channel,li2002simplified,coleri2002channel,kim2005qrd,simon2013iterative}, minimum mean-square error (MMSE), least-squares (LS) and Kalman-based channel estimators are proposed, which cost different computation complexity is not the scope of this paper. Thus, we ignore the channel estimation complexity and only consider the complexity on ICI reduction.

In Alg. \ref{alg:ICIReductionAlgLOS} and \ref{alg:originalAlgRician}, the complexity of generating ${\rm{\mathbf{S}}}_r^L$ and ${\rm{\mathbf{S}}}_r^N$ is ${\rm{O}}(T_x T_y L N^2)$. In addition, the approximate frequency equalization is carried out by simply multiplexing the transpose of $\rm{\mathbf{S}}_r^L$ or $\rm{\mathbf{S}}_r$, which cost ${\rm{O}}(N^2)$. Thus, the overall ICI reduction computational complexity is ${\rm{O}}(T_x T_y L N^2)$ and no additional bandwidth cost is needed.

%In AWGN scenarios, we assume that the large-scale path loss can be predicted by the known relative positions between RRUs and train. Thus, the frequency equalization matrix ${\rm{\mathbf{S}}}_r^L$ can be directly calculated by \eqref{equ:expressionMatFormFlatFading} and no additional bandwidth is needed in Alg. \ref{alg:ICIReductionAlgLOS}.

%For Rician scenarios, the small-scale channel fading needs to be estimated. When small-scale channel fading is acquired, the frequency equalization matrix ${\rm{\mathbf{S}}}_r$ can also be directly calculated via \eqref{equ:DFTICIExpression}. Thus, only the bandwidth cost on estimating small-scale channel fadings is needed in Alg. \ref{alg:originalAlgRician}.

%% 数值结果：
\section{Numerical Results}\label{Sec:NumRes}

In this section, numerical results are presented to show the validity of our theoretical analysis and provide more insights on the effectiveness of proposed ICI reduction algorithm. By absorbing the transmit power into large-scale path loss, which is modeled via Okumura-Hata model\cite{medeisis2000use}, the received signal power can be normalized (i.e. $\sigma_x^2 = 1$) and the channel gains can be represented as $G(d) = B - 10\alpha {\rm{log}}(d)$, where $\alpha$ is related to the height of receive antennas and defined as 3.8. The propagation distance is denoted by $d$. The parameter $B$ controls the receive signal power and has been set as $126$. $d_h$ and $d_v$ in Fig. \ref{Fig:CoverageModel} are defined as 500m and 100m, respectively. Referencing the simulation parameters in \cite{lu2014precoding}, the OFDM sampling duration $T_s$ is 71$\mu s$ and 2.4GHz, respectively. The total subcarrier number $K=1024$.

\subsection{SIR and ASQ Performance versus $\omega_D$}\label{Sec:SimuSIRAndASQ}

We summary the maximum and minimum SIRs versus different $\omega_D$ in table \ref{tab:SIRAndMSs}. To illustrate the overall SIR performance of proposed method, we also give the theoretical SIRs calculated by \eqref{equ:SIRApprox} and the simulated SIRs with and without ICI reduction at specific $\omega_D = 0.08$ in Fig. \ref{Fig:VariousSIR}, that is, the velocity is 500 km/h. The black vertical lines denote the positions of RRUs. The propagation channels are assumed to be LOS (i.e. K $\rightarrow +\infty$). It can be seen that the calculated theoretical maximum and minimum SIRs with \eqref{equ:SIRApprox} are 51.75dB and 35.01dB, corresponding to simulated SIR 51.61dB and 35.22dB, respectively. This confirmed that our theoretical upper- and lower-bounds are tight and with high accuracy to its real values. In table \ref{tab:SIRAndMSs}, the minimum SIRs with proposed ICI reduction method are significantly higher than that without ICI reduction. In addition, at most of the positions in Fig. \ref{Fig:SIRsInLOSScenarios}, the resulted SIR via proposed ICI reduction method is observably higher than the SIR without ICI reduction. That is, for such a interference limited OFDM system, one can conclude that proposed ICI reduction method can effectively increase system capacity at most positions on the railway and provide promising mobile service enhancement.

%% mobile service in one service duration：
\begin{figure}[htbp]
\centering
\includegraphics[width=0.65\textwidth]{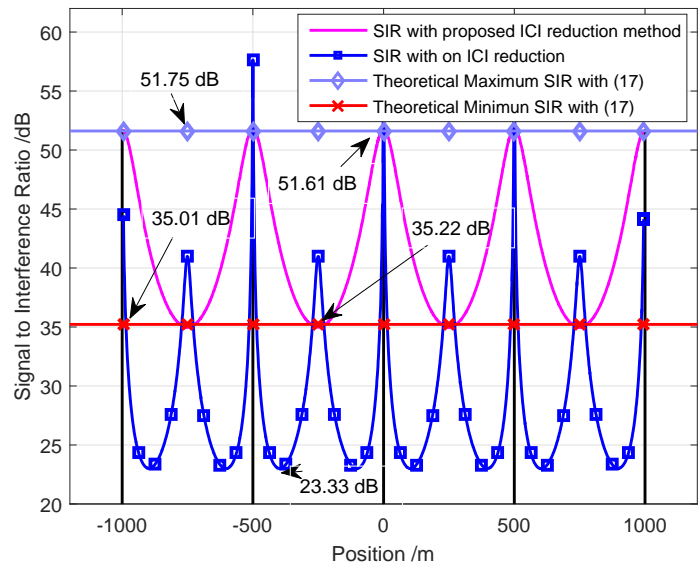}
\caption{The experimental SIRs Versus different Positions on 500km/h without scattering/reflecting and MIMO regime is $4 \times 4$. The x-label represents the position of receive antenna and the y-label represents corresponding SIRs with and without ICI reductions. The horizontal lines are theoretical maximum and minimal SIRs via \eqref{equ:SIRApprox}.}\label{Fig:SIRsInLOSScenarios}
\end{figure}

Also, the simulated and theoretical MSs versus different $\omega_D$ are summarized in table \ref{tab:SIRAndMSs}. In addition, for specific $\omega_D=0.05$ and 0.08, the theoretical ASQ with \eqref{equ:ApproxMobileService} and the ASQ with and without ICI reduction are depicted in Fig. \ref{Fig:MobileService}, corresponding to 300km/h and 500km/h, respectively. It can be observed that the approximation in \eqref{equ:ApproxMobileService} at discussed velocities is tight and the ASQ can be precisely predicted with \eqref{equ:ApproxMobileService}. Observe the ASQ without ICI reduction in table \ref{tab:SIRAndMSs} and Fig. \ref{Fig:MobileService}, one can find that the ICI increases significantly with respect to mobility and results in a sharp loss of ASQ. By contrast, our proposed ICI mitigation method can get almost the same performance with the ICI completely removed mode in terms of the mobile service in HSR scenarios with $\omega_D=0.05$. For HSR scenarios with $\omega_D=0.08 \sim 0.2$, the MS loss with ICI reduction varies from $13\%$ to $34\%$ compared with ICI completely removed. However, the MS loss without ICI reduction varies from $39\%$ to $55\%$. That is, proposed method can efficiently retain the MS of system, especially for the HSR scenarios with $\omega_D=0.05 \sim 0.08$, corresponding to $300 \rm{km/h} \sim 500 \rm{km/h}$.

%% mobile service in one service duration
\begin{figure}[htbp]
\centering
\includegraphics[width=0.65\textwidth]{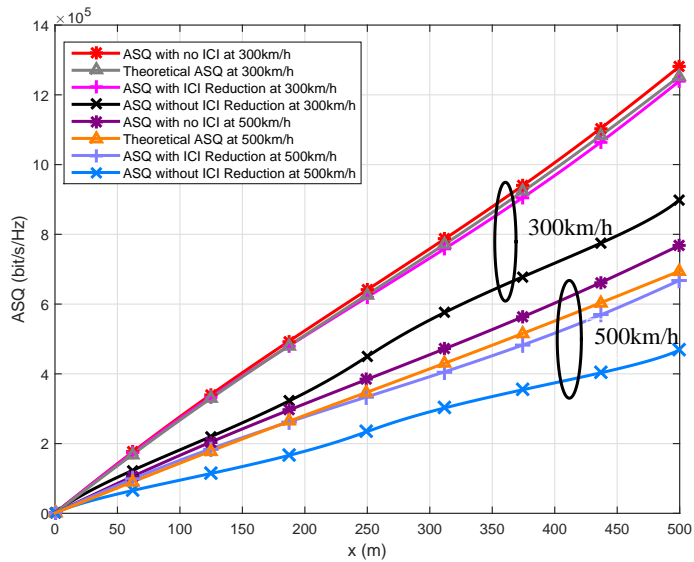}
\caption{The  accumulative service quantity (ASQ) over one subcarrier, where the velocities are 300km/h and 500km/h, corresponding to $\omega_D=0.05$ and 0.08, respectively.} \label{Fig:MobileService}
\end{figure}

\begin{table}[htbp]
 \caption{SIRs (\textnormal{shown in dB}) and MSs (\textnormal{bit/s/Hz, $\times 10^5$}) versus $\omega_D$}\label{tab:SIRAndMSs}
 \centering
 \begin{tabular}{cccccc}
  \toprule
  $\omega_D$ & 0.05 & 0.08  & 0.12  & 0.15 & 0.2 \\  
  \midrule
  Max reduced SIR & 60.6 & 51.73 & 45.96 & 41.63 & 38.2\\
  Min reduced SIR & 43.82 & 34.97 & 29.21 & 24.91 & 21.54\\
  Min SIR without reduction & 27.41 & 22.98 & 20.06 & 17.87 & 16.11\\
  \midrule
  MS with no ICI & 12.82 & 7.68 & 5.49 & 4.26 & 3.49\\
  MS with ICI reduction & 12.4 & 6.68 & 4.2 & 2.89 & 2.11\\
  Theoretical MS with \eqref{equ:ApproxMobileService} & 12.54 & 6.95 & 4.46 & 3.11 & 2.3\\
  MS without ICI reduction & 9.0 & 4.69 & 2.98 & 2.10 & 1.57\\
  \bottomrule
 \end{tabular}
\end{table}

\subsection{Effects of Antenna number and Rician K-factor}\label{Sec:SimuKAffects}

The effects of antenna numbers are summarized in table \ref{tab:AntennaNumbers}. It can be seen that the antenna numbers don't affect the expectation and variance of Alg. \ref{alg:ICIReductionAlgLOS} and Alg. \ref{alg:originalAlgRician}. This is because the signal and interference power in \eqref{equ:SIRExpression} and \eqref{equ:SIRRicianExpression} grows linearly with $T_xT_y$ simultaneously. Therefore, as the antenna number grows, the expectation and variance performance of SIR remains unchanged.

% 在一张表格中，给出几种不同天线数的性能表现，得出结论：天线数不影响（前面理论部分要提一句）
%% Effects of Antenna numbers
\begin{table}[htbp]
 \caption{Expectation and Variance versus Antenna Regimes\protect\\(\textnormal{shown in dB, iteration times = 1000})}\label{tab:AntennaNumbers}
 \centering
 \begin{tabular}{cccccc}
  \toprule
  \multicolumn{2}{c}{Antenna Regime} & $1 \times 1$ & $2 \times 2$  & $4 \times 4$  &  $8 \times 8$  \\  
  \midrule
  \multirow{2}{*}{Alg. \ref{alg:ICIReductionAlgLOS}} & $\mathbb{E}{\rm{(max(SIR_k))}}$ & 48.76 & 48.72 & 48.87 & 48.80\\
    & $\mathrm{Var}{\rm{(max(SIR_k))}}$ & 3.90 & 4.46 & 3.78 & 4.23 \\  
  \midrule
  \multirow{4}{*}{Alg. \ref{alg:originalAlgRician}} & $\mathbb{E}{\rm{(max(SIR_k))}}$ & 48.76 & 48.72 & 48.87 & 48.80\\
    & $\mathrm{Var}{\rm{(max(SIR_k))}}$ & 3.90 & 4.46 & 3.78 & 4.23 \\
    \cline{2-6}& $\mathbb{E}{\rm{(min(SIR_k))}}$ & 48.76 & 48.72 & 48.87 & 48.80\\
    & $\mathrm{Var}{\rm{(min(SIR_k))}}$ & 3.90 & 4.46 & 3.78 & 4.23 \\
  \bottomrule
 \end{tabular}
\end{table}

% 解释一下莱斯因子结论，绘制表格-->不同K因子下的均值、方差表现（两种算法）,给出Definition 1相对应的说明 
Fig. \ref{Fig:Simp20To40dB} and \ref{Fig:Orig10To30dB} demonstrate the SIR and variance performance of Alg. \ref{alg:ICIReductionAlgLOS} and Alg. \ref{alg:originalAlgRician} versus Rician factor. Besides, the theoretical and simulated expectations and variances are summarized in table \ref{tab:ExpAndVarVersusRicianKFactors}. The simulation parameter $\omega_D$ and antenna regime are 0.08 and $4 \times 4$, respectively. That is, the simulation results are shown in scenarios with velocity of 500km/h. Observe and compare the ICI reduction performance of Alg. \ref{alg:ICIReductionAlgLOS} and \ref{alg:originalAlgRician}, one can conclude that
\begin{enumerate}
\item In Alg. \ref{alg:ICIReductionAlgLOS}, the expectation and variance of maximum SIR increases and decreases with K, respectively. By contrast, even though the variances of maximum and minimum SIRs in Alg. \ref{alg:originalAlgRician} increase with K, the expectations remain stable.
\item Any channels with $K \geq 30\rm{dB}$ can be considered as AWGN channels. (In the scenarios with $K=30\rm{dB}$, which is marked in table \ref{tab:ExpAndVarVersusRicianKFactors}, ${\mathrm{Var}(\rm{SIR}_k)}/{\mathbb{E}(\rm{SIR}_k)} \approx 10\%$. According to definition \ref{def:ProperlyK}, one can claim that $K_p = 30\rm{dB}$.)
\item For arbitrary K, the variances of remained SIR in Alg. \ref{alg:originalAlgRician} are less than that in Alg. \ref{alg:ICIReductionAlgLOS}. When $K \leq 30\rm{dB}$, Alg. \ref{alg:originalAlgRician} can provide better ICI reduction performance. However, for $K \geq 30\rm{dB}$, adopt Alg. \ref{alg:ICIReductionAlgLOS} can provide satisfying performance.
\end{enumerate}

Finally, as shown in table \ref{tab:ExpAndVarVersusRicianKFactors}, except for scenarios with low Rician factors (i.e. 10dB), the theoretical expectations and variances are with high accuracy to the simulated ones. This is because the linear approximations in \eqref{equ:VarAppenISRSimMax}, \eqref{equ:VarAppenISROrigMax} and \eqref{equ:VarAppenISROrigMin} are with the assumption that K is sufficiently large. Notice the fact that, in most of the HSR scenarios, the scattering/reflecting is poor \cite{dong2015power,li2013channel}. The accuracy of theoretically approximated expectations and variances in \eqref{equ:SIRApprox} and lemma \ref{lem:VarianceRespectToRiceFactor} is sufficiently enough for HSR scenarios.

\begin{figure*}
\centering 
\subfigure[20dB] { \label{Fig:Simp20dB}     
\includegraphics[width=0.3\columnwidth]{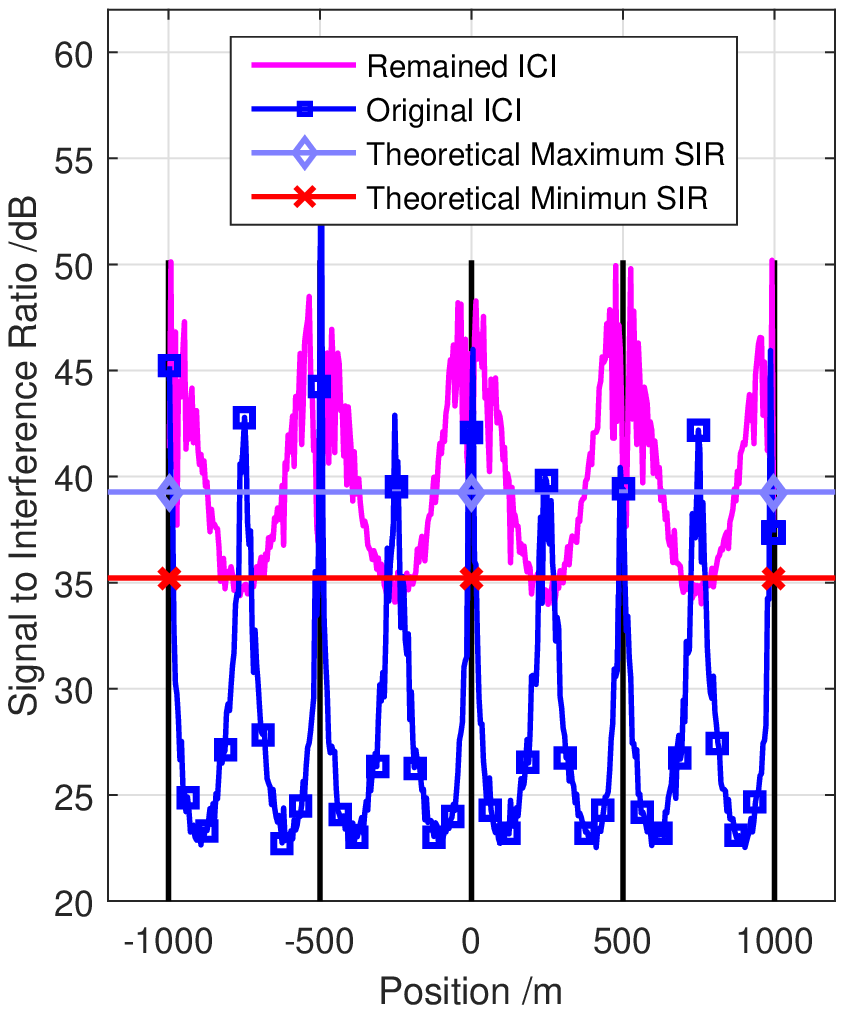}  
}
\subfigure[30dB] { \label{Fig:Simp30dB}     
\includegraphics[width=0.3\columnwidth]{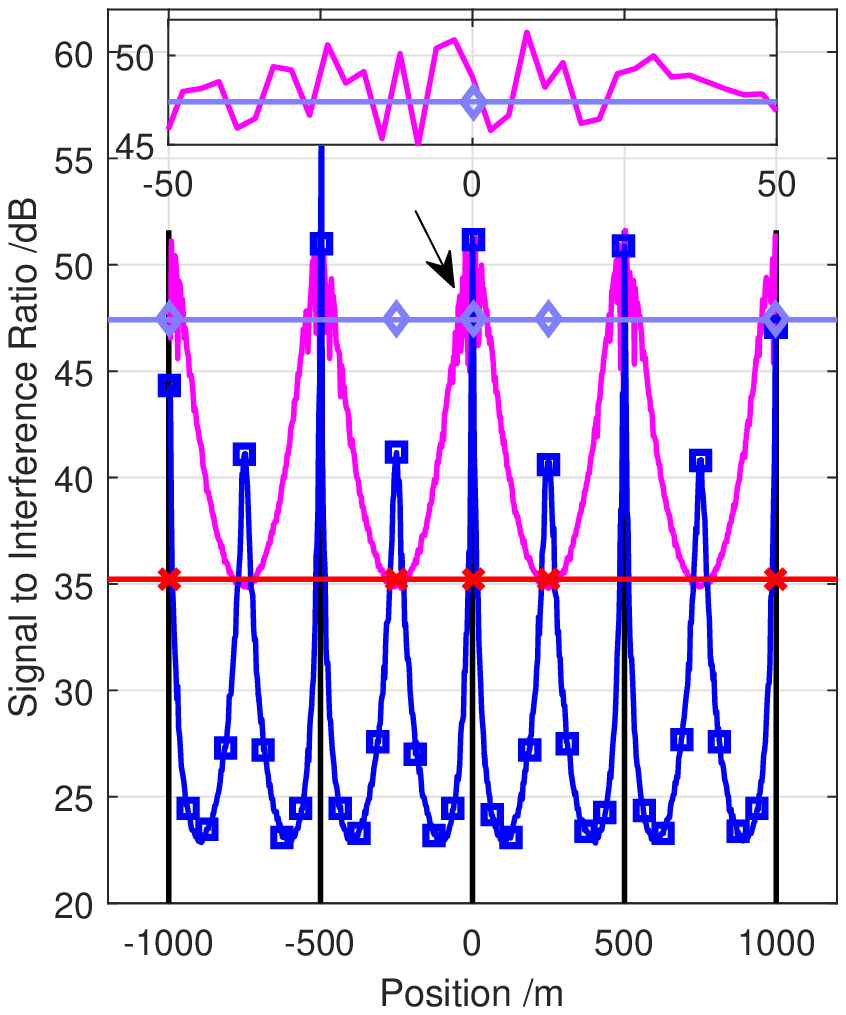}     
}
\subfigure[40dB] { \label{Fig:Simp40dB}     
\includegraphics[width=0.3\columnwidth]{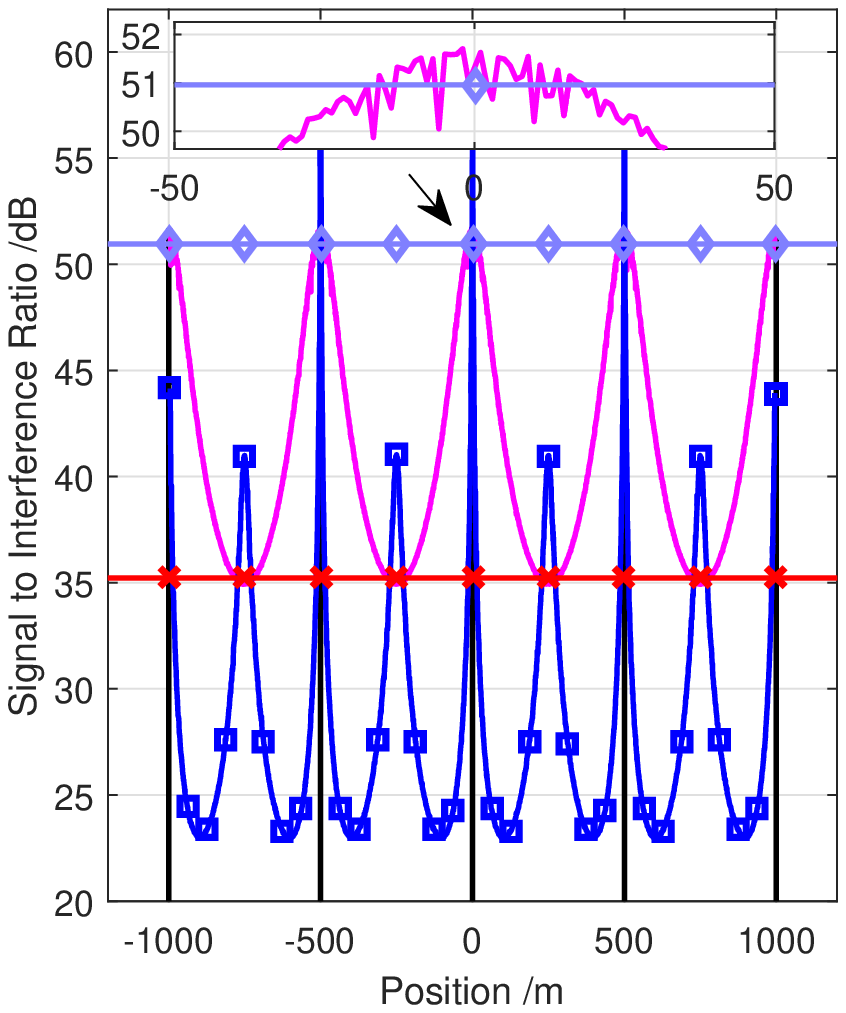}     
}
\caption{The ICI reduction effects with algorithm \ref{alg:ICIReductionAlgLOS}. The experimental SIRs versus different Positions on $\omega_D=0.08$ are shown different Rician factors, i.e. K = 20dB, 30dB and 40dB, respectively. The MIMO regime is 4$\times$4 and corresponding velocity is 500km/h.} \label{Fig:Simp20To40dB}
\end{figure*}

\begin{figure*}
\centering 
\subfigure[10dB] { \label{Fig:Orig10dB}     
\includegraphics[width=0.3\columnwidth]{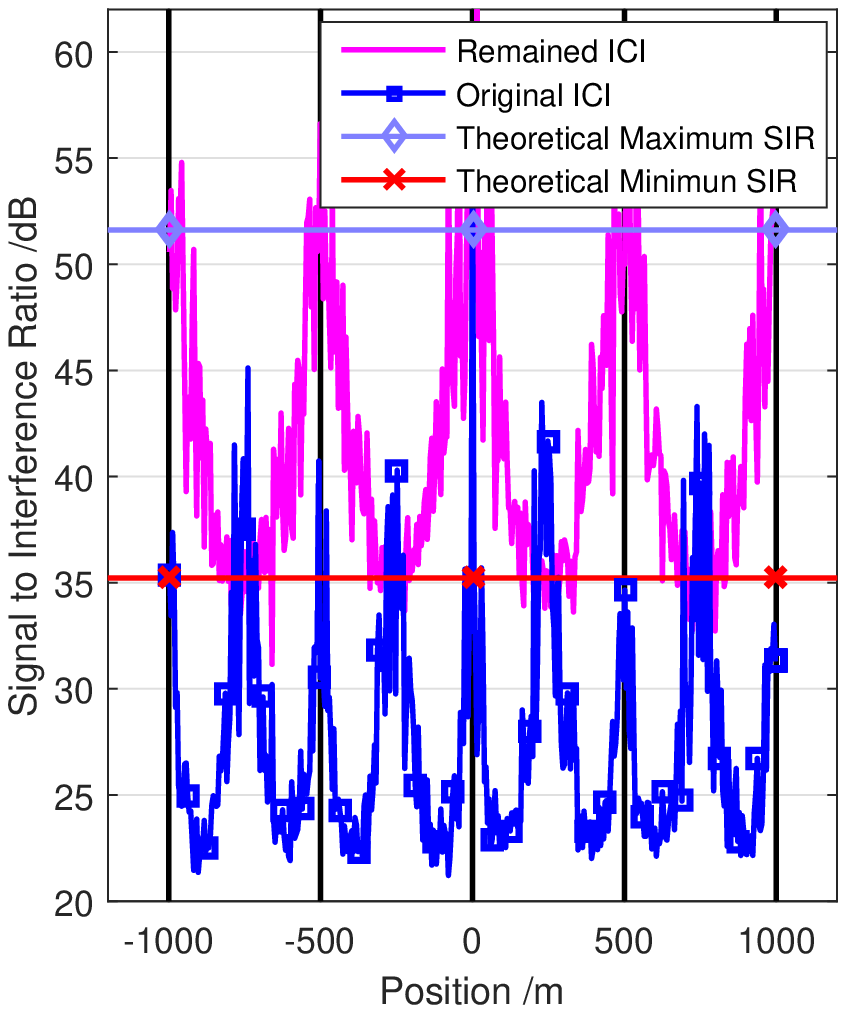}  
}    
\subfigure[20dB] { \label{Fig:Orig20dB}     
\includegraphics[width=0.3\columnwidth]{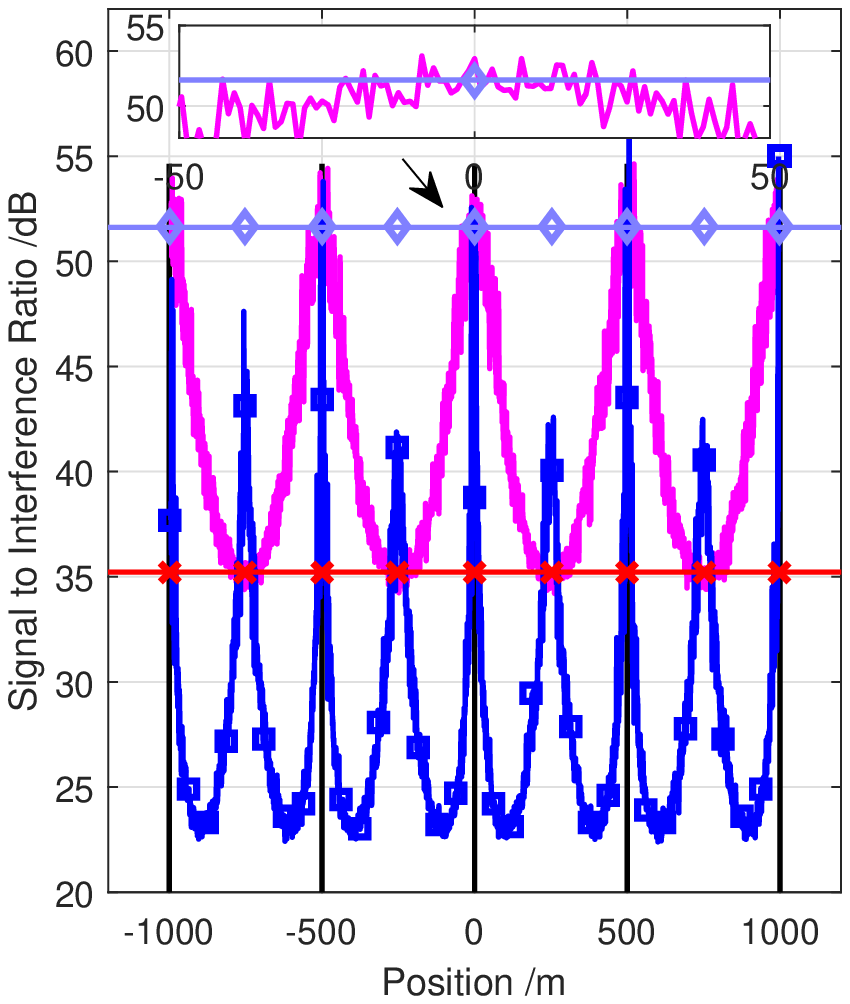}     
}    
\subfigure[30dB] { \label{Fig:Orig30dB}     
\includegraphics[width=0.3\columnwidth]{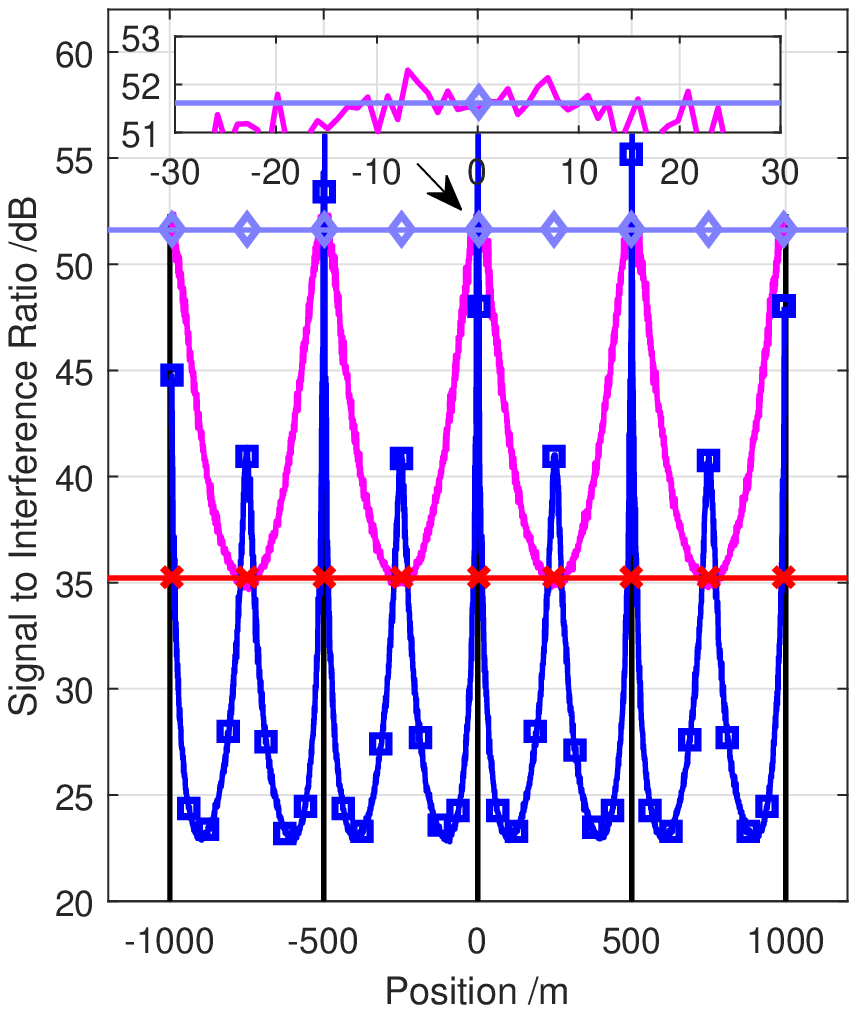}     
}      
\caption{The ICI reduction effects with algorithm \ref{alg:originalAlgRician}. The experimental SIRs versus different Positions on $\omega_D=0.08$ are shown different Rician factors, i.e. K = 10dB, 20dB and 30dB, respectively. The MIMO regime is 4$\times$4 and corresponding velocity is 500km/h.} \label{Fig:Orig10To30dB}    
\end{figure*}

\begin{table}[htbp]
 \caption{Expectation and Variance in Theory \textnormal{(Th)} and Simulation \textnormal{(Sim)} versus Rician factors \protect\\(\textnormal{shown in dB, iteration times = 1000})}\label{tab:ExpAndVarVersusRicianKFactors}
 \centering
 \begin{tabular}{ccccccc}
  \toprule
  \multicolumn{3}{c}{Rician Factors} & 10 & 20  & \multicolumn{1}{>{\columncolor{mygray}}c}{30}  &  40 \\ 
  \midrule
  \multirow{4}{*}{Alg. \ref{alg:ICIReductionAlgLOS}} & \multirow{2}{*}{$\mathbb{E}{\rm{(max(SIR_k))}}$} & Sim & - & 42.28 & \multicolumn{1}{>{\columncolor{mygray}}c}{48.62} & 51.25\\
  & & Th & - & 39.28 & \multicolumn{1}{>{\columncolor{mygray}}c}{47.42} & 50.96\\
  & \multirow{2}{*}{$\mathrm{Var}{\rm{(max(SIR_k))}}$} & Sim & - & 17.92 & \multicolumn{1}{>{\columncolor{mygray}}c}{4.16} & 0.19 \\  
  & & Th & - & 21.01 & \multicolumn{1}{>{\columncolor{mygray}}c}{3.98} & 0.18\\
  \midrule  
  \multirow{8}{*}{Alg. \ref{alg:originalAlgRician}} & \multirow{2}{*}{$\mathbb{E}{\rm{(max(SIR_k))}}$} & Sim & 51.68 & 51.72 & 51.73 & -\\
  & & Th & 51.61 & 51.61 & 51.61 & -\\
  & \multirow{2}{*}{$\mathrm{Var}{\rm{(max(SIR_k))}}$} & Sim & 8.19 & 0.33 & 0.03 & - \\  
  & & Th & 3.77 & 0.38 & 0.04 & -\\
  \cline{2-7}& \multirow{2}{*}{$\mathbb{E}{\rm{(min(SIR_k))}}$} & Sim & 35.0 & 35.0 & 35.0 & -\\
  & & Th & 35.22 & 35.22 & 35.22 & -\\
  & \multirow{2}{*}{$\mathrm{Var}{\rm{(min(SIR_k))}}$} & Sim & 1.20 & 0.12 & 0.01 & - \\
  & & Th & 1.89 & 0.19 & 0.02 & -\\
  \bottomrule
 \end{tabular}
\end{table}

%% 结论：
\section{Conclusion}\label{Sec:Conclusion}

This paper analyzed the ICI reduction problem in HSR scenarios for MIMO-OFDM downlinks with distributed antennas and prove that the total ICI matrix can be formulated as the weighting average of single ICI matrices at related SWGN/Rician downlinks, where the weighting coefficients are the corresponding channel gain factors. By analyzing the features of ICI matrix, the unitary property when $\omega_D \rightarrow 0$ is derived, based on which, an ICI reduction method without computational costly matrix inversion is proposed. The proposed ICI reduction method can retain the spectrum efficiency and achieve low complexity frequency equalization for real-time ICI reduction on the fast time varying HSR channels. In addition, the ICI reduction performance are analyzed via the derived SIR expectation and variance with respect to Rician $K$-factor. Numerical results show that our proposed ICI mitigation method can get almost the same performance with that obtained on the ICI completely removed mode in terms of the mobile service in the scenarios of HSR  with velocity of 300km/h.

It can be noted that with the location information, the single ICI matrices of LOS paths can be acquired and hence, the estimation of weighting coefficients and fadings corresponding to NLOS paths are the key of ICI mitigation in HSR scenarios, which may be taken into account for future works.

%More insights can be derived from the results of this paper, that is, for the ICI reduction problem of HSR scenarios with distributed antennas, the dominant factors are the channel gains $\rho_{r,t}^2$, relative moving speed and distance between transmit and receive antennas. The last two components can be easily acquired via positioning systems. Hence, the most important factor in HSR ICI reduction is the estimation of channel gain factors $\rho_{r,t}^2$.

% 致谢：
\section*{Acknowledgment}

This work was partly supported by the China Major State Basic Research Development Program (973 Program) No.2012CB316100(2), National Natural Science Foundation of China (NSFC) No.61171064 and NSFC No.61321061.

\begin{appendices}

% 给出多普勒扩展性能：
\section{}\label{Appen:DopplerVariance}
Because the AOAs and amplitudes of each NLOS paths is i.i.d. and $N_s$ is sufficiently large, according to the central limit theorem, $D(n-k)$ is complex Gaussian distributed and the variance is related to the power of the first and second term of \eqref{equ:DopplerSpreadSumExpress}. Denote the second term as $I(\theta)$, one can derive that

\begin{equation*}
\begin{split}
 E(I(\theta)^2) &= \int_{-\pi}^{\pi} I(\theta)^2 p(\theta) d\theta = \frac{1}{2\pi} \int_{-\pi}^{\pi} \left( \frac{sin(\pi (n-k+\omega_D cos(\theta_i)))}{N sin(\frac{\pi}{N}(n-k + \omega_D cos(\theta_i) ))} \right)^2 d\theta \\
&= \bigg\{
\begin{array}{cc}
\frac{1}{2\pi} \int_{-\pi}^{\pi} \left( \frac{sin(\pi \omega_D cos(\theta_i))}{ \pi \omega_D cos(\theta_i) }  \right)^2 d\theta  ,    & {n=k}\\
\frac{1}{2\pi} \int_{-\pi}^{\pi} \left( \frac{sin(\pi \omega_D cos(\theta_i))}{ N sin(\frac{\pi}{N}(n-k + \omega_D cos(\theta_i) ))}  \right)^2 d\theta  ,  & {n \neq k}.\\
\end{array}
\end{split}
\end{equation*}

Ignore the small $\omega_D$ when $(n-k) \neq 0$ and approximate $sin(\pi \omega_D cos(\theta_i))$ with $\pi \omega_D cos(\theta_i)$, the expectation of $I(\theta)^2$ is

\begin{equation*}
E(I(\theta)^2)= \bigg\{
\begin{array}{cc}
1 ,    & {n=k}\\
\frac{\omega_D^2}{2(n-k)^2}  ,  & {n \neq k}.\\
\end{array} .
\end{equation*}

Obviously, the mean of $D(n-k)$ is

\begin{equation*}
E(D(n-k)) = \sum_{i=1}^{N_s} E( a_i ) E(I(\theta)) = 0.
\end{equation*}

Hence, the variance of Doppler spread in Lemma \ref{lem:DopplerSpreadVariance} is

\begin{equation*}
\begin{split}
E(D(n-k)^2)= N_s E(a_i^2) E(I(\theta)^2) = \bigg\{
\begin{array}{cc}
\frac{1}{K+1} ,    & {n=k}\\
\frac{\omega_D^2}{2(K+1)(n-k)^2}  ,  & {n \neq k}.
\end{array} 
\end{split}
\end{equation*}

% 证明LOS正交性：
\section{}\label{Appen:LOSUnitary}
The proving of the unitary property of the transmission matrix $\mathbf{S}_r^{\rm{L}}$ can be split into two parts, i.e. the diagonal submatrices of $\mathbf{F}_r$ approach $(\sum_{t=1}^{L} \rho_{t,r} \mathbf{G})^T (\sum_{t=1}^{L} \rho_{t,r} \mathbf{G})$, denoted as $\bm{\beta}^T \bm{\beta}$, and the remaining interference caused by non-diagonal submatrices tends to be vanish as $\omega_D \rightarrow 0$.

From \eqref{equ:I_rk[n-k]}, we can see that $I_{r,t}[N+n-k] = I_{r,t}[n-k]$, $\lim\limits_{(n-k) \to \infty } {I_{r,t}[n-k]} = 0$ and

\begin{equation}\label{equ:SimI_rk[n-k]}
\begin{split}
\lim\limits_{N \to \infty } {I_{r,t}[n-k]} &= \frac{\rm{sin}\left( \pi \varepsilon_{r,t} \right) \rm{exp}\left(j \pi \varepsilon_{r,t} \right)}{ \pi(n+\varepsilon_{r,t}-k) } (-1)^{n-k}.
\end{split}
\end{equation}

Hence, the component $I_{r,t}[n-k]$ can be vanished for sufficiently large $(n-k)$. Substituting \eqref{equ:SimI_rk[n-k]} into \eqref{equ:I_mat} and vanishing the imaginary parts for tiny $\varepsilon_{r,t}$, we can derive the expression of $\mathbf{\Lambda}_{r}^{\rm{L}}(k,k+m)$ as

\begin{equation*} \label{equ:ExpressionOfF_r}
\begin{split}
&\mathbf{\Lambda}_r^{\rm{L}}(k,k+m) = \sum_{u=-(N-1)}^{N-1} \sum_{i=1}^{N_t} \sum_{j=1}^{N_t} \Big\{\\
& \frac{\rho_{r,i}\rho_{r,j} \rm{sin}(\pi (u - \varepsilon_{r,i})) \rm{sin}(\pi (u + m - \varepsilon_{r,j}))}{\pi^2 (u - \varepsilon_{r,i})(u + m - \varepsilon_{r,j})} \mathbf{G}^{T} \mathbf{G} \Big\}.
\end{split}
\end{equation*}

The diagonal submatrices of $\mathbf{\Lambda}_r^{\rm{L}}$, $m=0$, is

\begin{equation*} \label{equ:DiagonalApprox}
\begin{split}
\mathbf{\Lambda}_r^{\rm{L}}(k,k) = \sum_{i=1}^{N_t} \sum_{j=1}^{N_t} \frac{\rho_{r,i}\rho_{r,j}\rm{sin}(\pi \varepsilon_{r,i}) \rm{sin}(\pi \varepsilon_{r,j})}{\pi^2 \varepsilon_{r,i} \varepsilon_{r,j}} \mathbf{G}^T \mathbf{G} + 2 \sum_{u=1}^{N-1} \frac{1}{u^2} \sum_{i=1}^{N_t} \sum_{j=1}^{N_t} \frac{\rho_{r,i}\rho_{r,j}\rm{sin}(\pi \varepsilon_{r,i}) \rm{sin}(\pi \varepsilon_{r,j})}{\pi^2} \mathbf{G}^T \mathbf{G}.
\end{split}
\end{equation*}

Using the assumption that $N$ is sufficiently large and the conclusion of Riemann Zeta function, i.e. $\sum_{u=1}^{\infty} \frac{1}{u^k} = \zeta(k)$\cite{titchmarsh1986theory}, we can derive that $\lim\limits_{N \to \infty }\sum_{u=1}^{N-1} \frac{1}{u^2} = \zeta(2) = \frac{\pi^2}{6}$. Hence, \eqref{equ:ResDiagonalOfL} can be easily derived. When $\omega_D \rightarrow 0$, the first component of \eqref{equ:ResDiagonalOfL} can be approximated as $\sum_{i=1}^{N_t} \sum_{j=1}^{N_t} \rho_{r,i} \rho_{r,j} \mathbf{G}^T \mathbf{G} = \bm{\beta}^T \bm{\beta}$ and the second component can be vanished.

The non-diagonal submatrices of $\mathbf{\Lambda}_r^{\rm{L}}$, $m \neq 0$, is

\begin{equation} \label{equ:NonDiagonalApprox}
\begin{split}
\mathbf{\Lambda}_r^{\rm{L}}(k,k+m) = &\bigg\{ \sum_{i=1}^{N_t} \sum_{j=1}^{N_t} \frac{\rho_{r,i} \rho_{r,j} \rm{sin}(\pi \varepsilon_{r,i}) \rm{sin}(\pi \varepsilon_{r,j})}{\pi^2 \varepsilon_{r,i}(\varepsilon_{r,j} - m)} \mathbf{G}^T \mathbf{G} + \sum_{i=1}^{N_t} \sum_{j=1}^{N_t} \frac{\rho_{r,i} \rho_{r,j} \rm{sin}(\pi \varepsilon_{r,i}) sin(\pi \varepsilon_{r,j})}{\pi^2 (\varepsilon_{r,i} + m)\varepsilon_{r,j}} \mathbf{G}^T \mathbf{G}\\
&+ \sum_{u=-(N-1)}^{-m-1} \sum_{i=1}^{N_t} \sum_{j=1}^{N_t} \frac{\rho_{r,i} \rho_{r,j} sin(\pi \varepsilon_{r,i}) sin(\pi \varepsilon_{r,j})}{\pi^2(u-\varepsilon_{r,i})(u + m - \varepsilon_{r,j})} \mathbf{G}^T \mathbf{G}\\
&+ \sum_{u=-(m-1)}^{-1} \sum_{i=1}^{N_t} \sum_{j=1}^{N_t} \frac{\rho_{r,i} \rho_{r,j} sin(\pi \varepsilon_{r,i}) sin(\pi \varepsilon_{r,j})}{\pi^2(u-\varepsilon_{r,i})(u + m - \varepsilon_{r,j})} \mathbf{G}^T \mathbf{G}\\
&+ \sum_{u=1}^{N-1} \sum_{i=1}^{N_t} \sum_{j=1}^{N_t} \frac{\rho_{r,i} \rho_{r,j} sin(\pi \varepsilon_{r,i}) sin(\pi \varepsilon_{r,j})}{\pi^2(u-\varepsilon_{r,i})(u + m - \varepsilon_{r,j})} \mathbf{G}^T \mathbf{G} \bigg\}  (-1)^{m}.
\end{split}
\end{equation}

By approximating $u \pm \varepsilon_{r,i}$ with $u$ and adopting the assumption of $N \rightarrow \infty$, \eqref{equ:NonDiagonalApprox} can be simplified as \eqref{equ:ResNonDiagonalOfL}. The orders of the first and second component in \eqref{equ:ResNonDiagonalOfL} are $O(\omega_D)$ and the third component's order is $O(\omega_D^2)$, where $O(\cdot)$ is the order operator. As $\omega_D \rightarrow 0$, $O(\omega_D) + O(\omega_D^2) \rightarrow 0$. Hence, the ICI matrix $\mathbf{D}_r$ is approximately unitary when $\omega_D \rightarrow 0$.

As the unitary property declines with increasing velocity, the ICI reduction performance decreases and the remaining SIR can be expressed in \eqref{equ:SIRExpression}. Because $\mathbf{G}$ is the matrix full with elements 1, the ${\rm{SIR}}_k$ is determined by the coefficients in \eqref{equ:ResDiagonalOfL} and \eqref{equ:ResNonDiagonalOfL}. Thus, by dividing \eqref{equ:ResNonDiagonalOfL} by \eqref{equ:ResDiagonalOfL} and denoting the result as ${\rm{ISR}}_{k,m}$, the ISR caused by $(k+m)$-th subcarrier is

\begin{equation*}
\begin{split}
{\rm{ISR}}_{k,m} \thickapprox (\sum_{i=1}^{N_t}\sum_{j=1}^{N_t}\frac{2 \varepsilon_{r,i}\varepsilon_{r,j} + m \varepsilon_{r,j} - m \varepsilon_{r,i}}{\varepsilon_{r,i}\varepsilon_{r,j} + m \varepsilon_{r,j} - m \varepsilon_{r,i} - m^2})^2 \leq \frac{4 \omega_D^4 \cos^4(\theta_B)}{m^4},
\end{split}
\end{equation*}

Hence, the upper-bound of summed ${\rm{ISR}}_k$ is

\begin{equation*}
{\rm{ISR}}_k \leq 2\sum_{m=1}^{\infty} \frac{4 \omega_D^4 \cos^4(\theta_B)}{m^4} = 8 \omega_D^4 \cos^4(\theta_B) \zeta(4),%\frac{4 \pi^4 \omega_D^4}{45},
\end{equation*}
where $\sum_{m=1}^{\infty}\frac{1}{m^4} = \zeta(4) = \frac{\pi^4}{90}$ and the minimum SIR can be expressed as \eqref{equ:SIRApprox}. The minimum SIR corresponds to B in Fig. \ref{Fig:CoverageModel} and $\theta_B$ is the AOA of signals from the $(j+3)$-th RRU.

By contrast, A and C corresponds to the maximum SIR. Without loss of generality, A is considered. In this situation, $\theta_{j+1} = \pi-\theta_{j+3}$, $\theta_{j+2}=0$, $\rho_{j+1} = \rho_{j+3}$ and $\rho_{j+2} \gg \rho_{j+1}$ holds. For the concise of paper, denote $\theta_{j+3}$ as $\theta_{A}$. Due to the large-scale path loss, only signals from the $(j+1)$-th to the$(j+3)$-th RRUs are dominant and $L=3$. Substitute above equations into \eqref{equ:ResDiagonalOfL} and \eqref{equ:ResNonDiagonalOfL} and ignore the tiny parts, one can generate the lower-bound of remaining ISR at the $(k+m)$-th subcarrier as

\begin{equation*}
\begin{split}
&{\rm{ISR}}_{k,m} \geq \bigg \{ \frac{\frac{6\rho_{j+1}\rho_{j+2}\omega_D \cos(\theta_A) \rm{sin}(\pi \omega_D \cos(\theta_A)) }{ \pi m^2 }+\frac{8\rho_{j+1}^2\rm{sin}^2(\pi \omega_D \cos(\theta_A))}{\pi^2 m^2}}{\rho_{j+2}^2+4\frac{\rho_{j+1}^2 \rm{sin}^2(\pi\omega_D \cos(\theta_A))}{\pi^2 \omega_D^{2} \cos(\theta_A)}+4\frac{\rho_{j+1}^2 \rm{sin}^2(\pi\omega_D \cos(\theta_A))}{3}} \bigg\}^2 \thickapprox \left( \frac{6 \rho_{j+1} \omega_D^2 \cos^2(\theta_A) }{ \rho_{j+2} m^2} \right)^2.
\end{split}
\end{equation*}
The corresponding maximum ${\rm{SIR}}_k$ can be easily derived as \eqref{equ:SIRApprox} by substituting $\sum_{m=1}^{\infty}\frac{1}{m^4} = \zeta(4)$ into ${\rm{ISR}}_k = 2\sum_{m=1}^{\infty} {\rm{ISR}}_{k,m}$.

\section{}\label{Appen:RicianUnitary}
The proving of the unitary property of the transmission matrix $\mathbf{S}_r$ can also be split into two parts, i.e. the diagonal submatrices of $\mathbf{\Lambda}_r$ approach $(\sum_{t=1}^{L} \rho_{t,r} \frac{K}{K+1} \mathbf{G} + \frac{K}{K+1} \mathbf{R}_t)^T (\sum_{t=1}^{L} \rho_{t,r} \frac{K}{K+1} \mathbf{G} + \frac{K}{K+1} \mathbf{R}_t)$, denoted as $\bm{\zeta}^T \bm{\zeta}$, and the remaining interference caused by non-diagonal submatrices tends to be vanish as $\omega_D \rightarrow 0$.

Observing \eqref{equ:I_mat} and \eqref{equ:I_Dmat}, one can derive the expression of $\mathbf{\Lambda}_r^{\rm{L}}$ and $\Lambda_r^{LN}$ by simply replacing $\mathbf{G}$, $\frac{sin(\pi \varepsilon_{r,j})}{\pi (n + \varepsilon_{r,j} - k)}$ and $\frac{sin(\pi \varepsilon_{r,j})}{\pi \varepsilon_{r,j}}$ with $\mathbf{R}_i$, $\frac{\omega_D}{\sqrt{2}(n-k)}$ and 1, respectively. Thus, adopting the assumption of $N \rightarrow \infty$, the corresponding diagonal and non-diagonal submatrices can be easily expressed as \eqref{equ:ResDiagonalOfN}, \eqref{equ:ResNonDiagonalOfN}, \eqref{equ:ResDiagonalOfLN} and \eqref{equ:ResNonDiagonalOfLN}.

Vanish the components with orders $O(\omega_D)$ and $O(\omega_D^2)$ as $\omega_D \rightarrow 0$. The conclusion in Lemma \ref{lem:OriginalOrthogonalRician} can be easily proved.

Referencing the remaining SIR in \eqref{equ:SIRRicianExpression}, the ICI reduction performance of proposed method decreases with increasing velocity. Observe the structure of \eqref{equ:ResDiagonalOfL}, \eqref{equ:ResNonDiagonalOfL} and \eqref{equ:ResDiagonalOfN}-\eqref{equ:ResNonDiagonalOfLN}, one can claim that the SIR performance of Alg. \ref{alg:originalAlgRician} is similar with Alg. \ref{alg:ICIReductionAlgLOS}. Thus, the maximum and minimal $\rm{SIR}_k$ is achieved at A and B in Fig. \ref{Fig:CoverageModel}, respectively. 

Observing the expression of $\mathbf{\Lambda}_r^{\rm{LN}}(k,k+m)$ and considering the fact that $\mathbf{\Lambda}_r^{\rm{NL}} = (\mathbf{\Lambda}_r^{\rm{LN}})^T$, one can easily derive that 

\begin{equation}\label{equ:SumOfNLAndLN}
\begin{split}
\mathbf{\Lambda}_r^{\rm{NL}}(k,k+m) + \mathbf{\Lambda}_r^{\rm{LN}}(k,k+m) = (-1)^m \sum_{i=1}^{N_t} \sum_{j=1}^{N_t} \big\{2 \frac{\rho_i \rho_j sin(\pi\varepsilon_{r,i})\varepsilon_{r,i}}{\pi (\varepsilon_{r,i}^2 - m^2)} + 2\sqrt{2} \frac{\rho_{r,i} \rho_{r,j} \varepsilon_{r,i} \omega_D}{m^2} \big\} \mathbf{G}^T \mathbf{R}_j.
\end{split}
\end{equation}

When train is at A, the assumptions in Appendix \ref{Appen:LOSUnitary} holds as well. Substitute \eqref{equ:SumOfNLAndLN} into \eqref{equ:SIRRicianExpression} and vanishing the tiny parts, the lower-bound of remaining ${\rm{ISR}}_{k,m}$ and ${\rm{ISR}}_k$ can be expressed as \eqref{equ:ISRKMRicianMin} and \eqref{equ:ISRRicianMin}, which is shown at the top of next page. $\theta_A$ has the same definition as Appendix \ref{Appen:LOSUnitary}. Similarly, when train is at B, $\theta_{j+2} = \pi - \theta_{j+3}$, $\rho_{j+2} = \rho_{j+3}$ and $L = 2$. Denote $\theta_{j+3}$ as $\theta_B$, the corresponding upper-bounds of ${\rm{ISR}}_{k,m}$ and ${\rm{ISR}}_{k}$ can be written as \eqref{equ:ISRKMRicianMax} and \eqref{equ:ISRRicianMax}, respectively. Thus, the maximum and minimal ${\rm{SIR}}_{k}$ can be expressed as \eqref{equ:SIRApproxRicianMax} and \eqref{equ:SIRApproxRicianMin}.

To analysis the mean of ${\rm{SIR}}_{k}$, ignore the parts of $\mathbf{R}_{j+2}^T \mathbf{R}_{j+2}$ and $(\mathbf{R}_{j+1} + \mathbf{R}_{j+2})^T (\mathbf{R}_{j+1} + \mathbf{R}_{j+2})$ in \eqref{equ:SIRApproxRicianMax} and \eqref{equ:SIRApproxRicianMin} for sufficiently large $K$. In addition, because $\mathbf{R}_{j+1}$ and $\mathbf{R}_{j+2}$ are Gaussian distributed, the mean of signal and interference parts in \eqref{equ:SIRApproxRicianMax} and \eqref{equ:SIRApproxRicianMin} are the same as that in \eqref{equ:SIRExpression}. Thus, the average ${\rm{SIR}}_{k}$ can be also expressed as \eqref{equ:SIRApprox}.

\begin{figure*}[!t]
\normalsize

\begin{equation}\label{equ:ISRKMRicianMin}
\begin{split}
{\rm{ISR}}_{k,m} \geq \frac{||6K\frac{ \rho_{j+1}\rho_{j+2} \omega_D^2 \cos^2(\theta_A)}{ m^2 } \mathbf{G}^T \mathbf{G} - 4\sqrt{K}\frac{\rho_{j+1}\rho_{j+2}\omega_D^2 \cos^2(\theta_A)}{m^2} \mathbf{G}^T \mathbf{R}_{j+2} + 2 \frac{\rho_{j+2}^2 \omega_D^2}{m^2} \mathbf{R}_{j+2}^T \mathbf{R}_{j+2} ||_2^2}{||K \rho_{j+2}^2 \mathbf{G}^T \mathbf{G} + 2\sqrt{K} \rho_{j+2}^2 \mathbf{G}^T \mathbf{R}_{j+2} + \rho_{j+2}^2 \mathbf{R}_{j+2}^T \mathbf{R}_{j+2}||_2^2}
\end{split}
\end{equation}

\begin{equation}\label{equ:ISRRicianMin}
\begin{split}
&{\rm{ISR}}_k \geq \frac{ 2\zeta(4)||6K \rho_{j+1}\rho_{j+2}  {\omega_D}^2 \cos^2(\theta_A) \mathbf{G}^T \mathbf{G} - 4\sqrt{K} \rho_{j+1}\rho_{j+2}\omega_D^2 \cos^2(\theta_A) \mathbf{G}^T \mathbf{R}_{j+2} + 2 \rho_{j+2}^2 \omega_D^2 \mathbf{R}_{j+2}^T \mathbf{R}_{j+2} ||_2^2}{||K \rho_{j+2}^2 \mathbf{G}^T \mathbf{G} + 2\sqrt{K} \rho_{j+2}^2 \mathbf{G}^T \mathbf{R}_{j+2} + \rho_{j+2}^2 \mathbf{R}_{j+2}^T \mathbf{R}_{j+2}||_2^2}
\end{split}
\end{equation}

\begin{equation}\label{equ:ISRKMRicianMax}
\begin{split}
{\rm{ISR}}_{k,m} \leq \frac{||8K \frac{\rho_{j+2}^2 \omega_D^2 cos^2(\theta_B)}{m^2} \mathbf{G}^T \mathbf{G} + \frac{\rho_{j+2}^2 \omega_D^2}{m^2}(\mathbf{R}_{j+1} + \mathbf{R}_{j+2})^T (\mathbf{R}_{j+1} + \mathbf{R}_{j+2})||_2^2}{||4K\rho_{j+2}^2 \mathbf{G}^T \mathbf{G} + 2 \sqrt{K} \rho_{j+2}^2 \mathbf{G}^T(\mathbf{R}_{j+1} + \mathbf{R}_{j+2}) + \rho_{j+2}^2 (\mathbf{R}_{j+1} + \mathbf{R}_{j+2})^T (\mathbf{R}_{j+1} + \mathbf{R}_{j+2}) ||_2^2}
\end{split}
\end{equation}

\begin{equation}\label{equ:ISRRicianMax}
\begin{split}
{\rm{ISR}}_k \leq \frac{2\zeta(4) ||8K \rho_{j+2}^2 \omega_D^2 cos^2(\theta_B) \mathbf{G}^T \mathbf{G} + \rho_{j+2}^2 \omega_D^2 (\mathbf{R}_{j+1} + \mathbf{R}_{j+2})^T (\mathbf{R}_{j+1} + \mathbf{R}_{j+2})||_2^2}{||4K\rho_{j+2}^2 \mathbf{G}^T \mathbf{G} + 2 \sqrt{K} \rho_{j+2}^2 \mathbf{G}^T(\mathbf{R}_{j+1} + \mathbf{R}_{j+2}) + \rho_{j+2}^2 (\mathbf{R}_{j+1} + \mathbf{R}_{j+2})^T (\mathbf{R}_{j+1} + \mathbf{R}_{j+2}) ||_2^2}
\end{split}
\end{equation}
\hrulefill
\vspace*{4pt}
\end{figure*}

\section{}\label{Appen:Variance}
In the ICI reduction method aiming at LOS scenarios, the SIR at the $k$-th subcarrier can be immediately expressed as
\begin{equation}\label{equ:SIRLOSInRician}
{\rm{SIR}}_k = \frac{|| \mathbf{\Lambda}_r(k,k) ||_2^2}{\sum_{i=1,\neq k}^{N} || \mathbf{\Lambda}_r(k,i) ||_{2}^{2}},
\end{equation}
where 
\begin{equation*}
\mathbf{\Lambda}_r(k,i) = \sqrt{\frac{K}{K+1}} \mathbf{\Lambda}_r^{\rm{L}}(k,i) + \sqrt{\frac{1}{K+1}}\mathbf{\Lambda}_r^{\rm{LN}}(k,i).
\end{equation*}
Substitute \eqref{equ:DiagonalApprox}, \eqref{equ:NonDiagonalApprox}, \eqref{equ:ResDiagonalOfLN} and \eqref{equ:ResNonDiagonalOfLN} into \eqref{equ:SIRLOSInRician}, the approximate lower-bound of remaining ISR at the $(k+m)$-th subcarrier can be expressed by

\begin{equation*}
\begin{split}
&{\rm{ISR}}_{k,m} \geq \frac{||6\sqrt{K}\frac{ \rho_{j+1}\rho_{j+2} \omega_D^2 \cos^2(\theta_A)}{ m^2 } \mathbf{G}^T \mathbf{G} + \frac{\rho_{j+2}^2\omega_D}{\sqrt{2}m} \mathbf{G}^T \mathbf{R}_{j+2}||_2^2}{||\sqrt{K} \rho_{j+2}^2 \mathbf{G}^T \mathbf{G} + \rho_{j+2}^2 \mathbf{G}^T \mathbf{R}_{j+2} ||_2^2}.
\end{split}
\end{equation*}
In LOS scenarios, the Rician K-factor is sufficiently large. Thus, the random part in denominator can be vanished and ${\rm{ISR}}_{k,m}$ can be simplified as

\begin{equation*}
\begin{split}
&{\rm{ISR}}_{k,m} \geq \frac{36K\frac{\omega_D^4 \cos^4(\theta_A)}{m^4} + \frac{\psi \omega_D^2}{2m^2} r_{j+2}^2 }{K \psi},
\end{split}
\end{equation*}
where $r_{j+2}$ is the representative element of $\mathbf{R}_{j+2}$. By substituting $\sum_{m=1}^{\infty}\frac{1}{m^4} = \zeta(4)$ and $\sum_{m=1}^{\infty}\frac{1}{m^2} = \zeta(2)$ into ${\rm{ISR}}_k = 2\sum_{m=1}^{\infty} {\rm{ISR}}_{k,m}$, the corresponding minimal ${\rm{ISR}}_k$ can be easily derived as

\begin{equation}\label{equ:MeanAppenISRSimMax}
\begin{split}
&{\rm{ISR}}_{k} \geq \frac{72K \omega_D^4 \cos^4(\theta_A)\zeta(4) + \psi \omega_D^2 \zeta(2) r_{j+2}^2}{K \psi}.
\end{split}
\end{equation}

Furthermore, let us evaluate the variance of maximum ${\rm{SIR}}_k$ in log scale. Consider the fact that Rician K-factor is sufficiently large, the ${\rm{SIR}}_k$ can be linearly approximated by

\begin{equation}\label{equ:VarAppenISRSimMax}
\begin{split}
{\rm{SIR}}_k \approx &10\log_{10}\bigg( \frac{72\omega_D^4 \cos^4(\theta_A) \zeta(4)}{\psi} \bigg) + 2\sqrt{2} \log_{10}\bigg(1 + \frac{\psi \zeta(2)}{29\omega_D^2 \cos^4(\theta_A) \zeta(4) K}\bigg)r_{j+2}^2.
\end{split}
\end{equation}
Because $r_{j+2}$ is Gaussian distributed, $\mathbb{E}(r_{j+2}^2)=1$ and $\mathrm{Var}(r_{j+2}^2)=1$. Thus, the expectation and variance of maximum ${\rm{SIR}}_k$ can be easily derived from \eqref{equ:MeanAppenISRSimMax} and \eqref{equ:VarAppenISRSimMax}, respectively.

In the ICI reduction method aiming at Rician scenarios, the maximum and minimum ${\rm{SIR}}_k$ can be expressed as \eqref{equ:SIRApproxRicianMax} and \eqref{equ:SIRApproxRicianMin}. Vanish the random parts in denominator and ignore the parts of $\mathbf{R}_{j+2}^T \mathbf{R}_{j+2}$ and $(\mathbf{R}_{j+1}+\mathbf{R}_{j+2})^T (\mathbf{R}_{j+1}+\mathbf{R}_{j+2})$, the maximum and minimum ${\rm{SIR}}_k$ can be linearly approximated by

\begin{equation}\label{equ:VarAppenISROrigMax}
\begin{split}
{\rm{max}} \left(\rm{SIR}_k\right) \approx 10\log_{10}&\left( \frac{\psi^2 }{72\zeta(4)\omega_D^4\cos^4(\theta_A)} \right) + \frac{20}{\sqrt{K}\ln(10)} r_{j+2} ~~~{\rm{dB}}^2
\end{split}
\end{equation}
and
\begin{equation}\label{equ:VarAppenISROrigMin}
\begin{split}
{\rm{min}} \left(\rm{SIR}_k\right) \approx 10\log_{10}&\left( \frac{1 }{4\zeta(4)\omega_D^4\cos^4(\theta_B)} \right) + \frac{10}{\sqrt{K}\ln(10)} \Big(r_{j+1} + r_{j+2}\Big) ~{\rm{dB}}^2.
\end{split}
\end{equation}
Because $r_{j+1}$ and $r_{j+2}$ are Gaussian distributed and independent, $\mathrm{Var}(r_{j+1} + r_{j+2}) = 2$. Thus, the variance of maximum and minimal ${\rm{SIR}}_k$ can be easily derived from \eqref{equ:VarAppenISROrigMax} and \eqref{equ:VarAppenISROrigMin}, respectively.

\end{appendices}

%% 生成引用文献：
\bibliography{bibfile}

%\end{IEEEbiography}


\end{document}